%% file: main2.tex
\documentclass[11pt]{article}
\usepackage{authblk}
\usepackage{amssymb}
\usepackage{amsthm}
\usepackage{mathtools}
\usepackage{amsfonts}
\usepackage{amsmath}
\usepackage{tikz}
\usetikzlibrary{decorations.pathreplacing}
\usetikzlibrary{calc}

\title{Point-to-set Principle and Constructive Dimension Faithfulness} 

\author[1]{Satyadev Nandakumar}
\author[2]{Subin Pulari}
\author[1]{Akhil S}
\affil[1]{
	Department of Computer Science and Engineering\\
	Indian Institute of Technology Kanpur,
	Kanpur, Uttar Pradesh, India.
}
\affil[1]{\{satyadev,akhis\}@cse.iitk.ac.in}

\affil[2]{
	Universit\'e de Bordeaux, CNRS, Bordeaux INP, LaBRI, UMR 5800, F-33400, Talence, France
}
\affil[2]{subin.pulari@labri.fr}

\usepackage{amsmath,amsthm,amssymb}
\usepackage[margin=1in]{geometry}
\usepackage[colorlinks=true]{hyperref}

\newtheorem{theorem}{Theorem} 
\newtheorem{lemma}{Lemma} 
 
\newtheorem{corollary}{Corollary}
\newtheorem{definition}{Definition}
\newtheorem*{theorem*}{Theorem}
\newtheorem{conjecture}{Conjecture}

\include{commands2}

\date {}
\begin{document}

\maketitle

\begin{abstract}
	Hausdorff $\Phi$-dimension is a notion  of Hausdorff dimension developed using a restricted class of coverings of a set. We introduce an effective version of Hausdorff $\Phi$-dimension, which we call constructive $\Phi$-dimension.
	We prove a point-to-set principle for $\Phi$-dimension. We also provide a characterization of 
	constructive $\Phi$-dimension using Kolmogorov complexity and $s$-gales.

	Finally, we apply these tools to study \emph{faithfulness} of coverings $\Phi$. A family of coverings $\Phi$  is said to be
	faithful to Hausdorff dimension if the $\Phi$-dimension and
	Hausdorff dimension coincide for every set. Similarly, $\Phi$ is said to be
	faithful to constructive dimension if the constructive $\Phi$-dimension and
	constructive dimension coincide for every set.
	
	We derive the necessary and sufficient
	conditions for the constructive dimension faithfulness of the
	coverings generated by the Cantor series expansion,  based on the  
	terms of the expansion.
	Using the point-to-set principle for Cantor coverings, we show that the same condition characterises Hausdorff dimension faithfulness of  Cantor coverings, thereby giving an information theoretic proof of the result by Albeverio, Ivanenko, Lebid, and Torbin
	\cite{Albeverio2020}. 
	
	We investigate the question of weather the notions of faithfulness at Hausdorff and constructive levels are equivalent.
	Using a new
	technique for the construction of sequences satisfying a certain
	Kolmogorov complexity condition, we show that the notions of
	``faithfulness'' of Cantor coverings at the Hausdorff and
	constructive levels are equivalent, independent of the log-limit condition.

\end{abstract}

\section{Introduction}

\subsection{Faithfulness in dimension}

In the study of randomness and information, an important concept is
the preservation of randomness across multiple representations of the
same object. Martin-L\"of randomness and computable randomness, for
example, are preserved among different base-$b$ representations of the
same real (see Downey and Hirschfeldt \cite{Downey10}, Nies
\cite{Nies2009}, Staiger \cite{Staiger02}) and when we convert from
the base-$b$ expansion to the continued fraction expansion
(\cite{Nandakumar2008a}, \cite{NANDAKUMAR2022104876},
\cite{Nandakumar2020c}).

A quantification of this notion is whether the \emph{rate} of
information is preserved across multiple
representations. This rate is studied using a constructive analogue of Hausdorff dimension called constructive dimension \cite{Lutz2003b}\cite{Mayordomo02}.
Hitchcock and Mayordomo \cite{HitchMayor13} show that  
constructive
dimension is preserved across base-$b$ representations. However, in a
recent work, Akhil, Nandakumar and Vishnoi \cite{Nandakumar2023b} show
that the rate of information is not preserved across all
representations. In particular, they show that constructive dimension
is \emph{not} preserved when we convert from base-$b$ representation
to continued fraction representation of the same real.

This raises the following question: \emph{Under which settings is the
	effective rate of information - \emph{i.e.} constructive dimension -
	preserved when we change representations of the same real?}. Since
constructive dimension is a constructive analogue of Hausdorff
dimension, this question is a constructive analogue of the concept of
``faithfulness'' of Hausdorff dimension.

A family of covering sets $\Phi$ over a metric space $\X$ is ``faithful'' to Hausdorff dimension if the  dimension of every
set $\F \subseteq \X$ defined using covers constructed using $\Phi$, called the Hausdorff
$\Phi$-dimension, coincides with the Hausdorff dimension of $\F$.   Faithfulness is well-studied as determining the
Hausdorff dimension of a set is often a difficult problem, and
faithful coverings help simplify the calculation. This
notion is introduced in a work of Besicovitch \cite{Besicovitch52},
which shows that the class of dyadic intervals is faithful for
Hausdorff dimension. Rogers and Taylor
\cite{Rogers70} further develop the idea to show that all covering
families generated by comparable net measures are faithful for
Hausdorff dimension. This implies that the class of covers generated
by the base $b$ expansion of reals for any $b \in \N \setminus \{1\}$
is faithful for Hausdorff dimension. However, not all coverings are
faithful for Hausdorff dimension. A natural example is the continued
fraction representation, which is not faithful for Hausdorff dimension
\cite{PeresTorbin}. Faithfulness of Hausdorff dimension has then
been studied in various classes of coverings \cite{Albeverio2020},
\cite{Albeverio05}, \cite{IbragimTorbin15}, \cite{PeresTorbin}.

\subsection{Constructive Dimension Faithfulness}

In this work, we introduce a constructive analogue of Hausdorff
$\Phi$-dimension which we call constructive $\Phi$-dimension. A family of covering sets $\Phi$ over a metric space $\X$ is ``faithful'' to constructive dimension if the  constructive $\Phi$-dimension of every set $\F \subseteq \X$ coincides with the constructive dimension of $\F$. 
Mayordomo and Hitchcock \cite{HitchMayor13} show that all
base-$b$ representations of reals, which are faithful for Hausdorff dimension, are also faithful for constructive dimension. On
the other hand, Nandakumar, Akhil, and Vishnoi \cite{Nandakumar2023b} show that the
continued fraction expansion, which is not faithful for Hausdorff
dimension is also not faithful for constructive dimension.
This raises the natural question: \emph{Are faithfulness with respect to
	Hausdorff dimension and faithfulness with respect to constructive
	dimension equivalent notions?} A positive answer to this question
implies that Hausdorff dimension faithfulness, a geometric notion, can be
studied using the tools from information theory. Conversely, the
faithfulness results of Hausdorff dimension can help us understand the
settings under which constructive dimension is invariant for
\emph{every} individual real.

To study this question, we use a $\Phi$-dimensional analogue
of the point-to-set principle.
The point-to-set principle introduced by J. Lutz and
N. Lutz \cite{Lutz2018} relates the Hausdorff dimension of a set of $n$-dimensional reals with
the constructive dimensions of points in the set, relative to a minimizing
oracle. This theorem has been instrumental in answering open questions
in classical fractal geometry using the theory of computing (See
\cite{Lutz2020}, \cite{LutzStull18}, \cite{LutzStull20},
\cite{NeilLutz21}, \cite{Lutz2018}). Mayordomo, Lutz, and Lutz
\cite{Lutz2022} extend this work to arbitrary separable metric spaces.

\subsection{Our Results}

In this work, we formalise the notion of Hausdorff $\Phi$-dimension and constructive $\Phi$-dimension.   We prove the point-to-set principle for
$\Phi$-dimension (Theorem \ref{thm:pspforPhidim}).
We show that the notion of constructive $\Phi$-dimension is robust, by giving equivalent characterisations using Kolmogorov complexity  (Theorem \ref{thm:KCCharnofCPhidim}) and $\Phi$-$s$-supergales (Theorem \ref{thm:CdimPhiGales}).

The Cantor series
expansion, introduced by Georg Cantor \cite{Cantor1932}, uses a sequence of natural numbers $Q = \{n_k\}_{k \in \N}$ as the terms of representation. It generalises the notion of base-$b$ representation. Whereas
base-$b$ representation use
exponentials with respect to a fixed $b$, 
$\{b^n\}_{n \in \N}$, the Cantor series representation $Q = \{n_k\}_{k \in \N}$ uses factorials $\{n_1 \cdot n_2 \dots n_k\}_{k \in \N}$ as
the basis for representation. 

Using Kolmogorov complexity, we derive a log-limit condition on the terms of the Cantor series expansion that characterizes the constructive dimension faithfulness of the associated constructive Cantor coverings (Theorem \ref{thm:EffFaithCondn}).  We also characterise the Hausdorff dimension faithfulness of Cantor coverings using the same
log-limit condition. Thereby we give an information theoretic proof of the result by Albeverio, Ivanenko, Lebid and Torbin  \cite{Albeverio2020} (Theorem~\ref{thm:AlbIvankio2}).

We develop a
combinatorial construction of sequences having Kolmogorov complexities
that grow at the same rate as a given sequence relative to any given
oracle (Theorem \ref{thm:KCChasinglemma}). This new combinatorial
construction may be of independent interest in the study of
randomness.

We then study the problem of equivalence of faithfulness notions at the constructive and Hausdorff levels.
Using the combinatorial construction and the point-to-set principle for $\Phi$-dimension, we show (independently of the log-limit condition) that for constructive Cantor coverings, the notions of
constructive faithfulness and Hausdorff dimension faithfulness are
equivalent (Theorem \ref{thm:equivalenceOfFaithfulness}).

\subsection{Organisation of the paper}

Section~\ref{sec:Prelim} contains the preliminaries, including the
notation and the necessary background on Hausdorff dimension and
constructive dimension.
The paper can be conceptually divided into two parts.

\medskip
\noindent\textbf{Part I: Constructive $\Phi$-dimension and its properties.}
This part comprises Sections~\ref{sec:PhiDim}--\ref{sec:EffPhi} and
develops the theory of constructive $\Phi$-dimension.
\begin{itemize}
	\item In Section~\ref{sec:PhiDim}, we introduce Hausdorff
	$\Phi$-dimension. 
	
	\item Section~\ref{sec:PSP} establishes a point-to-set principle for
	$\Phi$-dimension (Theorem~\ref{thm:pspforPhidim}). 
	
	\item In Section~\ref{sec:EffPhi}, we introduce constructive
	$\Phi$-dimension, defined via effective $\Phi$-null covers.
	
	We also show robustness by giving equivalent
	characterisations in terms of Kolmogorov complexity
	(Theorem~\ref{thm:KCCharnofCPhidim}) and $\Phi$-$s$-supergales
	(Theorem~\ref{thm:CdimPhiGales}).
\end{itemize}

\medskip
\noindent\textbf{Part II: Cantor covering dimension and faithfulness.}
In the second part, we apply the general theory developed in Part~I to
study the faithfulness of Cantor coverings.
\begin{itemize}
	\item In Section~\ref{sec:CantorDim}, we introduce constructive Cantor covering
	dimension. Using the results from Part~I, we
	derive a Kolmogorov complexity characterisation and establish a
	point-to-set principle for Cantor covering dimension.
	
	\item In Section~\ref{sec:CantorFaith}, we introduce the notions of
	Hausdorff and constructive dimension faithfulness. For constructive Cantor
	coverings, we characterise constructive dimension faithfulness via a
	log-limit condition on the Cantor series terms
	(Theorem~\ref{thm:EffFaithCondn}).
	
	\item In Section \ref{sec:HausdFaith}, we characterise the Hausdorff dimension faithfulness of Cantor coverings using the
	log-limit condition. Thereby we give an information theoretic proof of the result by Albeverio, Ivanenko, Lebid and Torbin  \cite{Albeverio2020} (Theorem~\ref{thm:AlbIvankio2}).
	
	\item In Section~\ref{sec:EqlFaith}, we study the problem of equivalence of faithfulness at constructive and Hausdorff levels.
	We show, independently of the log-limit condition, that
	for computable Cantor coverings, faithfulness with respect to
	constructive dimension is equivalent to faithfulness with respect to
	Hausdorff dimension
	(Theorem~\ref{thm:equivalenceOfFaithfulness}). 
\end{itemize}

\section{Preliminaries} \label{sec:Prelim}
\subsection{Notation}
We use $\mathbb{X}$ to denote the metric space under consideration with metric $d(x,y)$.
For a set $ U \subseteq \X$, we use $|U|$ to denote the diameter of
$U$, that is $|U| = \sup_{x,y \in U} d(x,y)$. 

We use $\emptyset$ to denote the empty set, and we assume $|\emptyset| = 0$.  We call two sets $U$ and $V$ incomparable if $U \not\subseteq V$ and $V \not\subseteq U$.

We use $\Sigma$ to denote the binary alphabet $\{0,1\}$, $\Sigma^*$
represents the set of finite binary strings, and $\Sigma^\infty$
represents the set of infinite binary sequences. We use $\lambda$ to denote the empty string.
We use $\lvert x
\rvert$ to denote the length of a finite string $x \in \Sigma^*$. For
an infinite sequence $X =X_0X_1X_2\dots$, we use $X \restr n$
to denote the finite string consisting of the first $n$ symbols of
$X$. When $n \geq m$ we also use the notation $X[m,n]$ to denote the
substring $X_mX_{m+1}\dots X_n$ of $X \in \Sigma^\infty$. 
$\N^*$ denotes a finite sequence of natural numbers $[a_1, a_2, \dots, a_n]$. 

We use $\N$ to denote the set of natural numbers (starting from 1). 
We use $\Q$ to denote the set of rational numbers and $\R$ to denote the set of reals.

Given infinite sequences $X_1,  \dots, X_n$, we define the
\emph{interleaved sequence} $X_1 \oplus X_2 \oplus \dots
X_{n}$ to be the interleaved sequence $X = X_1[0]X_2[0] \dots X_n[0]X_1[1]\dots X_n[1]\dots$. 
 We call a set of strings $\Rho \subset \Sigma^*$  prefix-free if there are no two strings $\sigma,\tau \in \Rho$ such that $\sigma$ is a proper prefix of $\tau$. Given $n \in \N$ we use $[n]$ to denote $\{0,1,\dots n-1\}$. 


\subsection{Hausdorff Dimension} \label{sec:HausdorffDimension}

The following definitions are originally given by Hausdorff \cite{Hausdorff1919}. We take the definitions from Falconer \cite{Falc03}. 

\begin{definition}  [Hausdorff \cite{Hausdorff1919}]
	Given a set $\F \subseteq \R^n$. A 
	collection of sets $\{U_i\}_{i \in \N}$, where for each $i \in \N$, $U_i \subseteq \X$  is
	called a \emph{$\delta$-cover of $\F$} if for all
	$i \in \N$, $|U_i| \leq \delta$ and $\F \subseteq
	\bigcup_{i\in\N}{U_i}$.
\end{definition}

\begin{definition}  [Hausdorff \cite{Hausdorff1919}]
	Given an $\F \subseteq \X$, for any $s>0$, define
	$$
	\mathcal{H}^s_\delta(\F) = \inf\left\{\sum_i |U_i|^s :
	\{U_i\}_{i\in\N} \text{ is a } \delta\text{-cover of } \F
	\right\}.
	$$
\end{definition}

As $\delta$ decreases, the set of admissible $\delta$ covers decreases.
Hence $\mathcal{H}^s_\delta(\F)$ increases.

\begin{definition}  [Hausdorff \cite{Hausdorff1919}]
	For $s \in (0,\infty)$, the \emph{s-dimensional Hausdorff outer 
		measure} of $\F$ is defined as:
	$$\mathcal{H}^s(\F) = \lim\limits_{\delta \rightarrow 0^+}
	\; \mathcal{H}^s_\delta(\F).$$
\end{definition}


Observe that for
any $t>s$, if $\mathcal{H}^s(\F) < \infty$, then
$\mathcal{H}^t(\F) = 0$ (see Section 2.2 in \cite{Falc03}).

Finally, we have the following definition of Hausdorff dimension.

\begin{definition} [Hausdorff \cite{Hausdorff1919}]
	For any $\F \subset \X$, the \emph{Hausdorff dimension} of $\F$ is defined as:
	$$\dim(\F) = \inf \{s\geq 0 : \mathcal{H}^s(\F) = 0 \}
	.$$
\end{definition}

\subsection{Constructive dimension of sequences}

Lutz \cite{Lutz2003b} defines the notion of effective (equivalently,
constructive) dimension of an infinite binary sequence
using the notion of lower semicomputable $s$-gales. 

\begin{definition}[Lutz \cite{Lutz2003b}]
	For $s \in \Rplus$, a binary $s$-gale is a function $d: \Sigma^{*}
	\to \Rplus$ such that $d(\lambda) < \infty$ and for all $w \in
	\Sigma^*$, $$ 2^s \cdot d(w) = \sum_{i \in \{0,1\}} d(wi).$$
\end{definition}

\begin{definition}[Lutz \cite{Lutz2003b}]
The \emph{success set} of a binary $s$-gale $d$ is $$S^\infty (d) = \left\{ X \in
\Sigma^\infty : \limsup \limits_{n\to\infty} d(X \restr n) =
\infty\right\}.$$
\end{definition}

For $\mathcal{F} \subseteq \Sigma^\infty$, $\mathcal{G}(\mathcal{F})$
denotes the set of all $s \in \Rplus$ such that there exists a lower
semicomputable (Definition \ref{def:lowerscfn}) binary $s$-gale $d$ with $\mathcal{F} \subseteq S^\infty (d)$.

\begin{definition}	[Lutz \cite{Lutz2003b}, Hitchcock \cite{Hitchcock2003c}]
	The \emph{constructive dimension} or \emph{effective Hausdorff
		dimension} of $\mathcal{F} \subseteq \Sigma^\infty$ is $$\cdim(\mathcal{F}) =
	\inf \mathcal{G}(\mathcal{F}).$$ 
	
	The constructive dimension of a
	sequence $X \in \Sigma^\infty$ is $\cdim(X) = \cdim(\{X\})$.
\end{definition}

The Kolmogorov complexity of a finite binary string is defined as the length of the shortest program (from a fixed prefix-free set) which produces the string. 

\begin{definition}
	The	Kolmogorov complexity of $\sigma \in \Sigma^*$ is defined as $K(\sigma)= \min\limits_{\pi \in \Sigma^*} \left\{ \lvert \pi \rvert : U(\pi)=\sigma \right\}$,	
	where $U$ is a fixed universal prefix-free Turing machine.
\end{definition}

Mayordomo \cite{Mayordomo02} extends the result by Lutz \cite{Lutz03} to give the following 
Kolmogorov complexity characterization of constructive dimension of infinite binary sequences.

\begin{theorem}[Lutz \cite{Lutz03}, Mayordomo \cite{Mayordomo02}] \label{lem:KCCharnofCdim}
	For any $X \in \Sigma^\infty$, $$\cdim(X) =
	\liminf\limits_{n\to\infty} \frac{K(X \restr n)}{n}.$$
\end{theorem}

\subsubsection{Constructive Dimension in Euclidean space}

Lutz and Mayordomo \cite{Lutz2008a} define constructive dimension of a point $x \in \R^n$ in Euclidean space by considering the binary expansions of these points, and porting it over to the Cantor space $\Sigma^\infty$. They note that if one or more of the coordinates of $x$ have two binary expansions, $\cdim(x)$ is unaffected by how we choose between these binary expansions.

\begin{definition} [Lutz and 
	Mayordomo \cite{Lutz2008a}]
	For a real $x = (x_1,...x_n) \in \R^n$, let $S_1 \in \Sigma^\infty,\dots,S_n \in \Sigma^\infty$ respectively be one of the binary expansions of the fractional parts of each of the coordinates of $x$. Define $\mathrm{bin}(x) = S_1 \oplus S_2 \dots \oplus S_n$. 
\end{definition}

\begin{definition} [Lutz and 
	Mayordomo \cite{Lutz2008a}]
		For an $x \in \R^n$, define
		$$\cdim(x) = n \cdot \cdim(\mathrm{bin}(x)).$$
\end{definition}

 Lutz and 
Mayordomo \cite{Lutz2008a} gave the following 
Kolmogorov complexity characterization of constructive dimension of points in $\R^n$.

\begin{theorem}[Lutz and 
	Mayordomo \cite{Lutz2008a}] \label{thm:CdimusingK}
	For any $\x \in \R^n$, and $A \in \Sigma^\infty$ 
	$$\cdim^A(\x) =
	\liminf\limits_{r\to\infty} \frac{K^A_r(\x)}{r}.$$
	where $K_r^A(\x) = \min\limits_{q \in \Q^n}\{K^A(q) : d(x,q) < 2^{-r}\}$.  Here $d$ is the Euclidean metric in $R^n$.
\end{theorem}

\section{ Hausdorff $\Phi$-dimension }  \label{sec:PhiDim}

Hausdorff dimension is defined using the notion of $s$-dimensional outer measures, where a cover is taken as the of union of a collection of covering sets $\{U_i\}_{i \in \N}$. Here a covering set $U_i$ can be any arbitrary subset of the space (see Section \ref{sec:HausdorffDimension}).
We define the general notion of Hausdorff $\Phi$-dimension by
restricting the class of admissible covers to $\Phi$-covers, which are the union of sets from a family of covering sets $\Phi$. 

\subsection{Family of covering sets}

In this work, we consider a \emph{family of covering sets} which satisfy the properties given below.

\begin{definition} [Family of covering sets $\Phi$] \label{def:PhiOg} We
	consider a metric space $\X$. A countable
	family of sets $\Phi$, where for
	each $\U \in \Phi$, we have $U \subseteq \X$, is called a family of covering sets over $\X$ if the following property holds
	
	\begin{itemize}
		\item Fineness: For all $x \in \X$ and $\delta > 0$, there exists a $U \in \Phi$ such that $x \in \U$ and $|U| < \delta$. 
	\end{itemize}
\end{definition}

We now define the notion of a $\Phi$-cover of a set. 
\begin{definition} [$\Phi$-cover] \label{def:PhiCover}
	Let $\Phi$ be a family of covering sets over $\X$. A $\Phi$ - cover of a set $\F \subseteq \X$ is a countable collection of sets $\{U_i\}_{i\in\N} \subseteq \Phi$ such that $\{U_i\}_{i\in\N}$ covers $\F$, that is $\F \subseteq \bigcup\limits_{i\in\N} U_i$. 
	
\end{definition}

\subsection{ Hausdorff $\Phi$-dimension}
\label{sec:MeasureDimension}

We call a $\Phi$-cover of $\F$ a $\delta$-cover if the diameter of elements in the cover are less than $\delta$.

\begin{definition}
	Let $\Phi$ be a family of covering sets defined over
	$\X$. Given a set $\F \subseteq \X$, a 
	$\Phi$-cover $\{U_i\}_{i \in \N}$  of $\F$ is
	called a \emph{$\delta$-cover of $\F$ using $\Phi$} if for all
	$i \in \N$, $|U_i| \leq \delta$. 
\end{definition}

From the fineness property in Definition \ref{def:PhiOg}, we have that for any set $\F \subseteq \X$ and $\delta >0$, $\delta$-cover of $\F$ using $\Phi$ always exist. 

\begin{definition} 
	Given an $\F \subseteq \X$, for any $s>0$, we define
	\begin{align*}
		\mathcal{H}^s_\delta(\F, \Phi) = \inf\left\{\sum_i |U_i|^s :
		\{U_i\}_{i\in\N} \text{ is a } \delta\text{-cover of } \F
		\text{ using } \Phi\right\}.
	\end{align*}
\end{definition}

If for any $\delta$, $\delta$-covers of $\F$ using $\Phi$ do not exist, then $\mathcal{H}^s_\delta(\F, \Phi) = \infty$.

As $\delta$ decreases, the set of admissible $\delta$-covers using $\Phi$ decreases. Hence $\mathcal{H}^s_\delta(\F, \Phi)$ increases.

\begin{definition} 
	For $s \in (0,\infty)$, define the \emph{s-dimensional $\Phi$ outer 
		measure} of $\F$ as:
	$$\mathcal{H}^s(\F, \Phi) = \lim\limits_{\delta \rightarrow 0^+}
	\; \mathcal{H}^s_\delta(\F, \Phi).$$
\end{definition}


Observe that as with the case of classical Hausdorff dimension, for
any $t>s$, if $\mathcal{H}^s(\F, \Phi) < \infty$, then
$\mathcal{H}^t(\F, \Phi) = 0$ (see Section 2.2 in \cite{Falc03}).

Finally, we have the following definition of Hausdorff
$\Phi$-dimension.


\begin{definition} 
	For any $\F \subset \X$, the \emph{Hausdorff $\Phi$-dimension} of $\F$ is defined as:
	$$\dim_\Phi(\F) = \inf \{s\geq 0 : \mathcal{H}^s(\F, \Phi) = 0 \}
	.$$
\end{definition}

\section{Point-to-set principle for $\Phi$-dimension}  \label{sec:PSP}
J. Lutz and N. Lutz \cite{Lutz2018} introduce the point-to-set principle for Hausdorff dimension. 
This provides an information-theoretic characterization of Hausdorff dimension, and has been fruitful in solving problems in classical dimension using tools from information theory  \cite{Lutz2020}. We show a $\Phi$-dimensional analogue of the point-to-set principle.

\subsection{Kolmogorov Complexity of $\Phi$-coverings}
We define the notion of Kolmogorov complexity of a point $x$ at precision $r$ with respect to $\Phi$. We denote this using $K_r(x,\Phi)$.
Let $\delta : \N \to \Phi$ be a fixed bijective enumeration of the elements in $\Phi$.

\begin{definition}
	Given an $r \in \N$ and $\x \in \X$, define 
	$$K_r(\x, \Phi) =\min_{U \in \Phi} \{K(\delta^{-1}(U)) : \x \in \U \text{ and } |U| \leq 2^{-r}\}.$$
\end{definition}

\subsection{{Point-to-set principle for $\Phi$-dimension} }

We show that the Hausdorff $\Phi$-dimension of a set can be characterised using 
relative Kolomogorov complexities of points in the set.
The proof is an adaptation of the proof from \cite{Lutz2018}. The idea is to encode indices corresponding to the $s$-dimensional $\Phi$-null covers of a set and their sizes  into the minimising oracle $A$.

\begin{theorem}\label{thm:pspforPhidim}
	Let $\Phi$ be a family of  covering sets over $\X$. For all $\mathcal{F} \subseteq \X$, 
	$$\dim_\Phi(\mathcal{F}) = \min\limits_{A \subseteq \N} \; \sup\limits_{\x \in \mathcal{F}} \; \liminf_{r \rightarrow \infty} \frac{K_r^A(\x, \Phi)}{r}.$$
\end{theorem}
\begin{proof}
	
	By the definition of $\Phi$-dimension, for any rational $s > \dim_\Phi(\mathcal{F})$ and $k \in \mathbb{N}$, there exists a sequence of $\Phi$-covers $\{\U_i^{k,s}\}_{i \in \mathbb{N}}$ of $\mathcal{F}$ such that
	\(
	\sum_{i \in \mathbb{N}} |\U_i^{k,s}|^s \leq 2^{-k}.
	\)
	
	 Consider the oracle sequence $A \subseteq \mathbb{N}$ that, for each $i, k \in \mathbb{N}$ and rational $s > \dim_\Phi(\mathcal{F})$, encodes:
	\begin{itemize}
		\item $f(i,k,s) = \delta^{-1}(\U_i^{k,s})$, the index of the $i$-th cover set, and
		\item $g(i,k,s) = 2^{-r}$, where $r \in \mathbb{N}$ is the unique integer satisfying $2^{-r-1} < |\U_i^{k,s}| \leq 2^{-r}$.
	\end{itemize}

	We show that given oracle access to $A$, for any rational $s > \dim_\Phi(\mathcal{F})$ and	for any point $x \in \F$, ${K_r^A(\x, \Phi)} \leq rs + o(r)$ for infinitely many $r \in \N$.
	
	Given $m \in \mathbb{N}$, choose $k > ms$. Since $\sum_{i \in \mathbb{N}} |\U_i^{k,s}|^s \leq 2^{-k} \leq 2^{-ms}$, every set $U$ in the cover satisfies $|U| < 2^{-m}$.

	For any $x \in \F$, consider the following description of the set in $\{\U_i^{k,s}\}_{i \in \mathbb{N}}$ that contains $x$. We specify the diameter range (an $r$ such that $2^{-r-1} < |U| \leq 2^{-r}$) of the set, and an index $j$ within sets of that diameter range. 
	Since $\sum_{i \in \mathbb{N}} |\U_i^{k,s}|^s \leq 1$, for any fixed $r \in \mathbb{N}$ there are at most $2^{rs + s} $ sets $U$ in the cover satisfying $2^{-r-1} < |U|$. Hence each such set can be indexed using $rs + o(r)$ bits. 
	
	More precisely, consider the Turing machine $M$ with oracle access to $A$ that, on input $(r, j, s, k)$, proceeds as follows: iterate through all $i \in \mathbb{N}$ and use $A$ to check whether $g(i,k,s) = r$. Upon finding the $j$-th such index $i$, halt and output $f(i,k,s)$.
	The input $(r, j, s, k)$ is encoded using $rs + o(r)$ bits to specify $j$ and $o(r)$ bits to specify $r,s,k$. Therefore, for any $m \in \N$, we have that for some $r \geq m$, $K_r^A(x, \Phi) \leq rs + o(r)$.
	
	\medskip
	
	We now show the reverse direction of the inequality.
	For an $A \in \Sigma^\infty$, take a rational $s > \sup\limits_{\x \in \mathcal{F}} \; \liminf\limits_{r \rightarrow \infty} \frac{K_r^A(\x, \Phi)}{r}$. For all $x \in \F$, we have that for some $r \in \N$, there exists a $U \in \Phi$, with $|U| < 2^{-r}$ and $K^A(\delta^{-1}(U)) \leq sr$. Define a $\Phi$-cover $\UU$ as follows. For each $w \in \Sigma^*$, if $K^A(w) \leq rs$ for some $r \in \N$ such that $ |\delta(w)| \leq 2^{-r}$, then add $\delta(w)$ to $\UU$.
	It follows that $\F \subseteq \bigcup_{U \in \UU} U $.
	We also have $\sum_{U \in \UU_k} |U|^s \leq \sum_{w \in \Sigma^*} 2^{-K(w)} \leq 1$. The last inequality follows from the Kraft inequality. Therefore $\cdim_{\Phi}(\F) \leq s$.
\end{proof}

\subsection{Point-to-set principle for constructive dimension}

\begin{definition} [Dyadic Family of covers] \label{def:dyadic_family}
	Consider the Euclidean space $\X = \R^n$. The dyadic family of covers is the set of coverings $\Phi_B = \bigcup_{r\in\N}  \{[\frac{m_1}{2^r}, \frac{m_1+1}{2^r}] \times \dots \times [\frac{m_n}{2^r}, \frac{m_n+1}{2^r}] \}_{m_1,m_2 \dots ,m_n \in [2^r]}$. 
\end{definition}

Besicovitch \cite{Besicovitch52} gave the following $\Phi$-dimension characterization of Hausdorff dimension. 

\begin{lemma}[Besicovitch \cite{Besicovitch52}] \label{thm:dimeqdimB}
	For all $\F \subseteq \R^n$, we have
	$\dim(\F) = \dim_{\Phi_B}(\F).$
\end{lemma}

Similarly, we have the following $\Phi$-dimension characterization of Constructive dimension. 

\begin{lemma} [Lutz and Mayordomo \cite{Lutz2008a}]\label{thm:cdimeqcdimB}
	For all $\F \subseteq \R^n$, we have
	$\cdim(\F) = \cdim_{\Phi_B}(\F).$
\end{lemma}

From Theorem \ref{thm:pspforPhidim} for $\Phi_B$, we have the following point-to-set principle from \cite{Lutz2018} relating Hausdorff and constructive dimensions. 

\begin{corollary}[J. Lutz and N. Lutz  \cite{Lutz2018}]\label{thm:PSP}
	For all $\mathcal{F} \subseteq \R^n$, $$\dim(\mathcal{F}) =
	\min\limits_{A \subseteq \N} \; \sup\limits_{\x \in
		\mathcal{F}} \; \liminf\frac{K^A(X \restr n)}{n}$$
	where $X$ is a binary expansion of $x$.
\end{corollary}

%
%
%
%

\section{Effective $\Phi$-dimension} \label{sec:EffPhi}

We effectivise the notion of $\Phi$-dimension using effective $\Phi$-covers of a set. Let $\delta : \N \to \Phi$ be a fixed enumeration of the elements in $\Phi$.

\begin{definition} [Effective $\Phi$-cover] \label{def:EffNullCover}
	
	For an $s \in \Rplus$, we say that a set $\F \subseteq \X$ has an effective $s$-dimensional $\Phi$-null cover if and only if there exists a Turing machine $M : \N \times \N \to \N$ such that:
	
	\begin{itemize}
		\item For all $r \in \N$, $\sum\limits_i \; \lvert \delta (M(i,r))\rvert^s$ $ \leq 2^{-r}$ and
		\item For all $r \in \N$, $\bigcup\limits_i \;  \delta (M(i,r))$ $ \supseteq \F$.
		
	\end{itemize}
	
\end{definition}

Note: If for any $i,r \in \N$, if $M(i,r)$ does not halt, we assume $\delta((M(i,r))) = \emptyset$ in Definition \ref{def:EffNullCover}.



\begin{definition} \label{def:EffDimbyCovers}
	For any $\F \subset \X$, the \emph{effective $\Phi$-dimension} of $\F$ is defined as:
	$$\cdim_\Phi(\F) = \inf \{s\geq 0 : \text{ there exists an effective $s$-dimensional $\Phi$-null cover of } \F \}
	.$$
\end{definition}

Note that Definition $\ref{def:EffNullCover}$ can be relativised with respect to an oracle $A \in \Sigma^\infty$. This is done by replacing the Turing machine in Definition  $\ref{def:EffNullCover}$ with a Turing machine with oracle access to $A$. We then say that $\F$ has an effective $s$-dimensional $\Phi$-null cover $\cdim_\Phi^A(\F)$ relative to oracle $A$. 

Relativising Definition \ref{def:EffDimbyCovers}, 

\begin{definition} \label{def:EffRelDimbyCovers}
	For $\F \subseteq \X$ and $A \in \Sigma^\infty$, the \emph{effective $\Phi$-dimension} of $\F$ relative to $A$ is defined as:
	$$\cdim_\Phi^A(\F) = \inf \{s\geq 0 : \text{ there exists an effective $s$-dimensional $\Phi$-null cover of } \F \text { relative to $A$} \}.$$
\end{definition}

We now give equivalent characterisations of effective $\Phi$-dimension using Kolmogorov complexity and $\Phi$-$s$-supergales. This shows that the notion of effective $\Phi$-dimension is robust.

\subsection{Covering sets with computable diameters}

Let $\Phi$ be a family of covering sets (Definition \ref{def:PhiOg}).
Let $\delta : \N \to \Phi$ be a fixed bijective enumeration of all the elements in $\Phi$.

\begin{definition} [Family covering sets with computable diameters] \label{def:CompPhi}
	Let $\Phi$ be a family of covering sets defined over $\X$. We say $\Phi$ is a family of covering sets with computable diameters if for all $w \in \N$, $|\delta(w)|$ is computable.
\end{definition}

When we say $|\delta(w)|$ is computable, we mean there exists a total Turing machine $M$ that  on input $w \in \N$ and $r \in \N$, outputs an $\ell \in \Q$ such that $||\delta(w)| - \ell | \leq 2^{-r}$.

We first show that the constructive $\Phi$-dimension of a set is the supremum of constructive $\Phi$-dimensions of points in the set. 

\begin{theorem} \label{thm:CDimSetfromPoint}
	Let $\Phi$ be a family of covering sets with computable diameters   over $\X$. For any $\F \subseteq \X$, we have 
	$$\cdim_\Phi(\F) = \sup\limits_{\x \in \F} \cdim_\Phi(\x).$$
\end{theorem}

\begin{proof}
	Let $s$ be a rational such that $s > \cdim_\Phi(\F)$. There exists an effective $s$-dimensional $\Phi$-null cover for $\F$. Therefore, for all points $x \in \F$, there exists an effective $s$-dimensional $\Phi$-null cover for $x$. Hence $\sup\limits_{\x \in \F} \cdim_\Phi(\x) \leq s$.
	Therefore, $\cdim_\Phi(\F) \geq \sup_{\x \in \F} \cdim_\Phi(\x)$. 
	
	In the other direction, consider any rational $s > \sup_{\x \in \F} \cdim_\Phi(\x)$. We show that there exists an effective $s$-dimensional $\Phi$-null cover for $\F$. 
	
	We consider a universal effective $s$-dimensional $\Phi$-null cover in the same manner as in Theorem 6.2.5 in \cite{Downey10}. Let $M_0, M_1, \dots$ be an effective listing of all single argument Turing machines. Given an $r \in \N$, for each $i \in \N$, 
	enumerate the outputs of $M_i(r+i)$, till the point $\sum_{w \in M_i(r+i)} |\delta(w)|^s > 2^{-(r+i)}$. Let $\{S_n^{i,r}\}_{n \in \N} $ be the corresponding list of strings. Consider $U_r = \bigcup_{i} \{S_n^{i,r}\}_{n \in \N}$.  We have $\sum_{w \in U_r} |\delta(w)|^s \leq 2^{-r}$. 
	
	We now show that $\cap_{r \in \N} \; U_r$ is the universal effective $s$-dimensional $\Phi$-null cover. For a $\mathcal{W} \subseteq \X$, let $M_j(r)$ be a machine that enumerates the effective $s$-dimensional $\Phi$-null covers of $\mathcal{W}$. As for all $r \in \N$, $\mathcal{W} \subseteq \bigcup_{w \in M_j(r+j)} \delta(w)$, we have that $\mathcal{W} \subseteq $ $\cap_{r \in \N} \; U_r$.

	Since for every point $\x \in \F$ there exists an effective $s$-dimensional $\Phi$-null cover of $x$,  the universal effective $s$-dimensional $\Phi$-null cover covers $\F$. Hence $\cdim_{\Phi}(\F) \leq s$.
\end{proof}

\subsection{ Effective $\Phi$-dimension using Kolmogorov complexity} 


We give an equivalent formulation of constructive $\Phi$-dimension 
of a point using Kolmogorov complexity. The proof is very similar to that of Theorem \ref{thm:pspforPhidim}.

\begin{theorem} \label{thm:KCCharnofCPhidim}
	Let  $\Phi$ be a family of covering sets with computable diameters, over the space $\X$. For any $\x \in \X$,
	
	$$\cdim_{\Phi}(\x) = \liminf_{r \rightarrow \infty} \frac{K_r(\x, \Phi)}{r}.$$
\end{theorem}
\begin{proof} 
		Take any any rational $s > \cdim_\Phi(x)$.  For any $k \in \mathbb{N}$, there exists a machine $M$ that outputs sequence of $\Phi$-covers $\{\U_i^{k,s}\}_{i \in \mathbb{N}}$ of $x$ such that
	\(
	\sum_{i \in \mathbb{N}} |\U_i^{k,s}|^s \leq 2^{-k}.
	\)
	
	The proof that $\liminf_{r \rightarrow \infty} \frac{K_r(\x, \Phi)}{r} \leq s$ follows as in the first part of the proof of Theorem~\ref{thm:pspforPhidim}, with the additional observation that the oracle $A$ can be computed. $f(.)$ can be computed using $M$, and $g(.)$ is computable as $\Phi$ is a family of covering sets with computable diameters.
	
	For the reverse direction, consider an $s > \liminf_{r \rightarrow \infty} \frac{K_r(\x, \Phi)}{r}$. Again, the proof that $ \cdim_\Phi(\mathcal{F}) \leq s$ follows as in the second part of the proof of Theorem~\ref{thm:pspforPhidim}, taking $\F = \{x\}$ and $A = \emptyset$. The $\Phi$-cover $\UU$ formed is effective as $\Phi$ is a $|\delta^{-1}(w)|$ is computable for all $w \in \Sigma^*$.
\end{proof}

\subsection{Effective $\Phi$-dimension using gales}

In this section, we give a characterisation of $\Phi$-dimension  using $\Phi$-$s$-supergales. A $\Phi$-$s$-supergale can be seen as a gambling strategy where the bets are placed on the covering sets from $\Phi$. The definitions in this subsection are adaptations from the setting of \emph{Nice covers} by Mayordomo \cite{Mayordomo2018}. 

Note crucially that our setting does not use the $c$-cover property in \cite{Mayordomo2018}, and therefore covers more classes $\Phi$. For example, the set of continued fraction covers is not a Nice cover. Mayordomo's goal was to generalise constructive dimension to general metric spaces and therefore Nice-covers are designed to preserve effective dimension. This is important as the aim of this work is to study faithfulness, that is if coverings preserve effective dimension or not.

\begin{definition} [Family of layerwise covering sets $\Phi$] \label{def:Phi} We
	consider the metric space $\X$. A
	countable family of sets $\Phi = \defPhi$, where for
	each $i \in \N ,n \in \N$, $\U_i^n \subseteq \X$, is called a \emph{layerwise family of covering sets} if it satisfies the following
	properties:
	
	\begin{itemize}	
		\item  Increasing Monotonicity: For every $n \in \N$, $U \in \{U_i^n\}_{i \in \N}$ and $m \leq n$, there is a unique $V \in \{U_i^m\}_{i \in \N}$ such that $U \subseteq V$. 
		
		\item Computable subsets: For every $n \in \N$, and $i \in \N$, the set $\{j \in \N : U_j^{n+1} \subseteq U_i^n\}$ is uniformly computable. 
		
		\item Fineness : Given any $\epsilon > 0$, and $\x \in \X$, there
		exists a $\U \in \Phi$ such that $|U| < \epsilon$ and $x \in \U$.
		
		\item Computable diameter: For all $i,n \in \N$, $|U_i^n|$ is computable.
	\end{itemize}
\end{definition}

When we say $|U_i^n|$ is computable, we mean there exists a total Turing machine $M$ that  on input $i,n,r \in \N$, outputs an $\ell \in \Q$ such that $||U_i^n| - \ell | \leq 2^{-r}$. The set $\{j \in \N : U_j^{n+1} \subseteq U_i^n\}$ is uniformly computable if there is a Turing machine which on input $i,j,n$ decides if $U_j^{n+1}  \subseteq U_i^n$.

Note that the number of sets $\{U_i^n\}$ in some level $n \in \N$ can also be finite and bounded by $m$. The definition still holds because in this case we take $U_j^n = \emptyset$ for $j > m$. From Increasing monotonicity property, it follows that all elements $\{U_i^n\}$ in a particular level $n \in \N$ are incomparable.

Note that the Increasing Monotonicity property is taken from \cite{Mayordomo2018}. For any $U$ in a layer, there is a unique $V$ in a layer above it, such that $U$ is a refinement of $V$ ($U \subseteq V$). 
This restriction enables us to meaningfully define $s$-supergales over $\Phi$. The covering sets in the same layer need not be disjoint. 
 
Note that the notion of layerwise family of covering sets (Definition \ref{def:Phi}) is a refinement of family of  covering sets with computable diameters (Definition \ref{def:CompPhi}). We take $\delta(<i,n>) = U_i^n$, where $<,>$ is the pairing function from $\N \times \N \to \N$. 

\begin{definition}[Mayordomo \cite{Mayordomo2018}] \label{def:superGale} 
	Let $\Phi = \defPhi$ be a family of layerwise covering sets from
	Definition \ref{def:Phi}. For $s \in \Rplus$, a $\Phi$-$s$-supergale
	is a function $d: \Phi \to \Rplus$ such that: 
	
	\begin{itemize}
		\item $\sum\limits_{U \in \{U_i^1\}_{i \in \N}} d(U) |U|^s < \infty$ and
		\item For all $n \in \N$ and all $\U \in \{\U_i^n\}_{i \in \N}$, the
		following condition holds:
		$$ d(\U)\cdot |\U|^s \geq \sum_{V \in  \{\U_i^{n+1}\}_{i \in \N}, V \subseteq U}  d(V) \cdot |V|^s\;.$$
	\end{itemize}
\end{definition}	

The following is the generalization of Kraft inequality for $s$-supergales from Mayordomo \cite{Mayordomo2018}.
\begin{lemma} [Generalisation of Kraft inequality \cite{Mayordomo2018}] \label{lem:kolmogorovInequality}
	Let $d$ be a $\Phi$-s-supergale. Then for every $\mathcal{E} \subseteq \Phi$ such that the sets in $\mathcal{E}$ are incomparable, we have that
	$$ \sum\limits_{V \in \mathcal{E}} d(V) |V|^s \leq \sum\limits_{U \in \{U_i^1\}_{i \in \N}} d(U) |U|^s.$$  
\end{lemma}

\begin{proof}
	
	Let $\sum\limits_{U \in \{U_i^1\}_{i \in \N}} d(U) |U|^s = c$ for some $c \in \R$. It suffices to show that the Lemma holds when sets in $\mathcal{E}$ are subsets of level sets till a finite level $N$. That is, there exists an $N \in \N$ such that for all $U \in \mathcal{E}$, $U \in \{U_i^m\}$ for some $m \leq N$.  We prove this using induction on $N$. The base case when $N = 1$ is immediate. Now assume the lemma holds for $N-1$, consider $\mathcal{E}_N = \{V \in \mathcal{E} \;\vert\; V \in \{U_i^N\}\}$ and $\mathcal{E}_{<N} = \mathcal{E} \setminus \mathcal{E}_N$. From the increasing monotonicity property in Definition \ref{def:Phi}, we have that for all $V \in \mathcal{E}_N$, $V \subseteq W$ for a unique $W \in \{U_i^{N-1}\}$. Let $\mathcal{E}_{N-1}'$ be the collection of such $W$'s. We have that elements in $\mathcal{E}_{<N} \cup \mathcal{E}_{N-1}'$ are incomparable and so $\sum\limits_{V \in \mathcal{E}_{<N} \cup \mathcal{E}_{N-1}'} d(V) |V|^s \leq c$. From the gale condition in Definition \ref{def:superGale},  for all $W \in \mathcal{E}_{N-1}'$, 
	$ d(W)  |W|^s \geq \sum_{V \in  \{\U_i^{N}\}, V \subseteq U}  d(V) |V|^s\;$. Therefore, $\sum\limits_{V \in \mathcal{E}} d(V) |V|^s \leq c$.		
\end{proof}

\begin{definition} [Mayordomo \cite{Mayordomo2018}] 	Let $\Phi = \defPhi$ be a family of layerwise covering sets.
	Given $\x \in \X$, a $\Phi$-representation of $\x$ is a sequence $(U_n)_{n \in \N}$ such that for each $n \in \N$, $U_n \in \{U_i^n\}_{i \in \N}$ and $\x \in \cap_n U_n$.
\end{definition}

Note that the same $\x$ can have multiple $\Phi$-representations. Given $\x \in \X$, let $\mathcal{R}(\x)$ be the set of $\Phi$-representations of $\x$.

\begin{definition} [Mayordomo \cite{Mayordomo2018}]
	A $\Phi$-s-supergale $d$ succeeds on $\x \in \X$ if there is a $(U_n)_{n \in \N} \in \mathcal{R}(\x)$ such that 
	$\limsup \limits_{n \rightarrow \infty} d(U_n) =\infty.$
\end{definition}

Equivalently, a $\Phi$-$s$-supergale $d$ succeeds on a point $\x \in \X$ iff for every $k \in \N$, there exists a $U \in \Phi$ such that $\x \in \U$ and $d(U) > 2^k$. 

\begin{definition} 
	The \emph{success set} of $d$ is $S^\infty (d) = \left\{ \x \in \X : d \text{ succeeds on \x} \right\}.$\end{definition}

We use constructive $\Phi$-$s$-supergales for constructive $\Phi$-dimension. For a  $\Phi$-$s$-gale $d$ to be constructive, we require the gale function $d$ to be lower semicomputable. Note that a lower semicomputable supergale actually takes as input $(i,n)$ where $i \in \N, n \in \N$ to place bets on $ U_i^n$. We omit this technicality in this paper and keep the domain of the gale as $\Phi$ for the sake of simplicity. 

\begin{definition} \label{def:lowerscfn}
	A function $d : \Phi \longrightarrow \Rplus$ is called
	\emph{lower semicomputable} if there exists a total computable
	function $\hat{d} : \Phi \times \mathbb{N} \longrightarrow
	\mathbb{Q} \cap \Rplus $ such that the following two
	conditions hold.
	\begin{itemize}
		\item \textbf{Monotonicity} : For all $\U \in \Phi$ and
		for all $n \in \mathbb{N}$, we have $ \hat{d}(\U,n) \leq
		\hat{d}(\U,n+1) \leq d(\U)$.
		\item \textbf{Convergence} : For all $U \in \Phi$,
		$\lim\limits_{n\to\infty} \hat{d}(\U,n) = d(\U)$.
	\end{itemize}
\end{definition}

\begin{definition}
	For $\mathcal{F} \subseteq \X$, let $\hat{\mathcal{G}}_{\Phi}(\mathcal{F})$
	denote the set of all $s \in \Rplus$ such that there exists a lower
	semicomputable $\Phi$-$s$-supergale $d$ with $\mathcal{F}
	\subseteq S^\infty (d)$.  
\end{definition}

We now give a $\Phi$-$s$-supergale characterisation of constructive $\Phi$-dimension for a family of layerwise covering sets $\Phi$.


\begin{lemma} \label{lem:CdimPhiGales1}
	Let $\Phi$ be a family of layerwise covering sets. For any  $\mathcal{F}
	\subseteq \X$, $$\cdim_{\Phi}(\mathcal{F}) \geq \inf
	\hat{\mathcal{G}}_{\Phi}(\mathcal{F}).$$
\end{lemma}

\begin{proof} 
	
	From Definition \ref{def:EffDimbyCovers},  for any rational $s > \cdim_\Phi(\mathcal{F})$ and $r \in \N$, there exists an effective sequence of covers $\{\U_i^r\}_{i \in \N}$ where each $\U_i^r \in \Phi$ such that $\sum_{i \in \N}|\U_i^r|^s \leq 2^{-r}$ and $\F \subseteq \bigcup_i \U_i^r$.
	
    Given $r \in \N$, consider the following $\Phi-$$s$-supergale $d_r : 
	\Phi \to \Rplus$  
	$$d_r(\U) = \frac{1}{|\U|^s} \;\left(\; { \sum_{{V} \in \{\U_i^r\} \; ; \; {V} \subseteq \U } |{V}|^s} \;\right).$$

	Note that $\sum\limits_{U \in \{U_i^1\}_{i \in \N}} d_r(U) |U|^s \leq  \sum_{i\in \N} |U_i^r|^s \leq 2^{-r}$. Since $|{V}|$ is computable, $d_r$ is lower semicomputable. It follows that for all $\U \in \{\U_i^r\}_{i\in \N}$, $d_r(\U) \geq 1$.  Finally define $d(\U) = \sum_{r=1}^{\infty}2^r \cdot d_{2r}(\U).$
	
	For any $r\in \N$, $\F \subseteq \bigcup_i \U_i^{2r}$. So, for all $x \in \F$ and  $r \in \N$, there exists a $\U \in \{\U_i^{2r}\}$ such that $x \in \U$. For that $\U$,  $d(\U) \geq 2^r d_{2r}(\U) \geq 2^r$. Therefore $\F \subseteq S^\infty(d)$ and so $\inf
	\hat{\mathcal{G}}_{\Phi}(\mathcal{F}) \leq s$.
\end{proof}

\begin{lemma} \label{lem:CdimPhiGales2}
	Let $\Phi$ be a family of layerwise covering sets. For any  $\mathcal{F}
	\subseteq \X$, $$\cdim_{\Phi}(\mathcal{F}) \leq \inf
	\hat{\mathcal{G}}_{\Phi}(\mathcal{F}).$$
\end{lemma}

\begin{proof} 
	Take any rational $s > \inf \hat{\mathcal{G}}_{\Phi}(\mathcal{F})$. There exists a lower semicomputable $\Phi$-$s$-supergale, say ${d}$ that succeeds on all $x \in \F$.
	
	Given an $r \in \N$, define 
	$\mathcal{U}_r = \{\U \in \Phi : 2^{-r-1} \leq |U| < 2^{-r} \text{ and } d(U) \geq 1\}.$
	
	In general, sets in $\mathcal{U}_r$ may not be incomparable. Therefore, we construct a new set $\mathcal{V}_r$ in the following manner. Once a new $U \in \mathcal{U}_r$ gets enumerated, we enumerate it into to $\mathcal{V}_r$ if and only if $U$ is incomparable with all $V$ currently enumerated into $\mathcal{V}_r$.
	Since the the sets in $\mathcal{V}_r$ are incomparable, from Lemma \ref{lem:kolmogorovInequality},  $\sum_{U \in \mathcal{V}_r} d(U) |U|^s < \infty$.
	
	Since $d(U) \geq 1$ and $|U| \geq 2^{-r-1}$, it follows that for some $c \in \N$, $\sum_{U \in \mathcal{V}_r} 2^{-rs} < c$. Therefore the number of elements in $\mathcal{V}_r$ is less than $2^{rs + o(r)}$.  
	
	Now consider the following set $\mathcal{W}_r$. Once a new $U \in \mathcal{V}_r$ gets enumerated, we find the largest $V \in \Phi$ such that $U \subseteq V$ and $2^{-r} \leq |V| < 2^{-r+1}$. We enumerate $V$ into $\mathcal{W}_r$. It follows that  the number of elements in $\mathcal{W}_r$ is less than $2^{rs + o(r)}$. Therefore given $r$, we need at most $rs+o(r)$ bits to index into a set within $\mathcal{W}_r$.	
	Therefore for all $U \in  \mathcal{W}_r$, $K(\delta^{-1}(U)) \leq rs + o(r)$.
	
	For any $x \in \F$, since ${d}$ succeeds on $x$, we have that for infinitely many $r \in \N$, $x \in U$ for some $U \in \mathcal{U}_r$. It follows that $U \subseteq V$ for some $V \in \mathcal{W}_r$.
	
	Therefore, for all $x \in \F$, for infinitely many $r \in \N$, $K_r(x, \Phi) \leq rs +o(r)$. 
	
	From Theorem \ref{thm:KCCharnofCPhidim}, we have $\cdim_{\Phi}(\F) \leq s$.
\end{proof}

From Lemma \ref{lem:CdimPhiGales1} and Lemma \ref{lem:CdimPhiGales2}, we have

\begin{theorem} \label{thm:CdimPhiGales}
	Let $\Phi$ be a family of layerwise covering sets. For any  $\mathcal{F}
	\subseteq \X$, $$\cdim_{\Phi}(\mathcal{F}) = \inf
	\hat{\mathcal{G}}_{\Phi}(\mathcal{F}).$$
\end{theorem}

\section{Constructive Cantor Covering Dimension} \label{sec:CantorDim}
\label{sec:equivalenceoffaithfulness}

In this section, we study the faithfulness of the family of coverings generated by computable Cantor series expansions.

\subsection{ Cantor coverings over the unit interval } 



Given a sequence $Q = \{n_k\}_{k \in \N}$ with $n_k \in \N \setminus
\{1\}$, the expression
$$x = \sum_{k = 1}^\infty \frac{\alpha_k}{n_1 \cdot n_2 \dots n_k}$$ where
$\alpha_k \in [n_k]$ is called the Cantor series expansion of the real
number $x \in [0,1]$  \cite{Cantor1932}.


\begin{definition}[Cantor coverings $\Phi_Q$]
	The class of Cantor coverings $\Phi_Q$ over the space $\X = [0,1]$ generated by the Cantor series expansion $Q = \{n_k\}_{k \in \N}$  is the set of 
	intervals $$\bigcup_{{k \in \N}} \{[\frac{m}{n_1 \cdot n_2 \dots n_k},
	\frac{m+1}{ n_1 \cdot n_2 \dots n_k}]\}_{m \in [ n_1 \cdot n_2 \dots
		n_k]}.$$ 
\end{definition}

For the class of Cantor coverings $\Phi_Q$ generated by the Cantor series expansion $Q = \{n_k\}_{k \in \N}$, we take $\delta(<k,m>) = [\frac{m-1}{ n_1 \cdot n_2 \dots n_k},
\frac{m}{n_1 \cdot n_2 \dots n_k}]$  when $m \leq n_1 \cdot n_2 \dots
n_k$. Otherwise $\delta(<k,m>) = \emptyset$. Here $<,>$ is the pairing function from $\N \times \N \to \N$.

\begin{definition}[Computable Cantor coverings]
	The Cantor series expansion $Q = \{n_k\}_{k \in \N}$ is said to be computable if there exists a machine that generates $n_k$ given $k$. We call the class of Cantor coverings $\Phi_Q$ generated by a computable Cantor series expansion $Q$ a class of \emph{computable Cantor coverings} over $\X = [0,1]$.
\end{definition}

\subsection{Point-to-set principle and Kolmogorov complexity characterization}

We use Theorem \ref{thm:KCCharnofCPhidim} to show a Kolmogorov complexity characterization of constructive $\Phi$-dimension for computable Cantor coverings, based on the terms appearing in the Cantor series expansion.


\begin{lemma}\label{lem:KCChanging}
	Let $\Phi_Q$ be a set of computable Cantor coverings 
	generated by $Q = \{n_k\}_{k \in \N}$. For all $x \in \X$ and $A \in \Sigma^\infty$,
	$$ \liminf_{r \rightarrow \infty} \frac{K_r^A(\x, \Phi)}{r} = \liminf\limits_{k\to\infty} \frac{K^A(X
		\restr m_k)}{m_k}$$ where $X$ is a binary
	expansion of $x$ and $m_k = \floor{\log_2(n_1 \cdot n_2 \dots
		n_k)}.$
\end{lemma}

\begin{proof}
	
	Consider a rational $s > \liminf_{r \rightarrow \infty} \frac{K_r^A(\x, \Phi)}{r}$. For infinitely many $r \in \N$, there exists a $U \in \Phi_Q$ such that $|U| \leq 2^{-r}$ and $x \in \U$ and $K^A(\delta^{-1}(U)) \leq sr$. Let $|U| = 1/(n_1 \cdot n_2 \dots n_k)$ for some $k \in \N$ and let $m_k = \floor{\log_2(n_1 \cdot n_2 \dots
		n_k)}$. We have that $2^{-(m_k+1)} < |U| \leq 2^{-m_k}$ and $m_k+1 > r$.
	
	We show that $K^A(X \restr m_k) \leq s\cdot m_k + o(m_k)$.  There are two dyadic intervals of size $2^{-m_k}$  that can cover $U$. We use the short description of $U$ to produce a short description of the dyadic interval of size $2^{-m_k}$ that contains $x$.  As $Q$ is computable, we can compute a $j$ such that $X \restr m_k = {2^{-m_k}} \cdot ({j+e})$ for some hardcoded $e \in \{0,1\}$. So we have
 $$K^A(X \restr m_k) \leq K^A(\delta^{-1}(U)) + o(m_k) \leq rs + o(m_k) \leq s\cdot m_k + o(m_k).$$ 
	
	Towards the other direction, assume that for infinitely many $k \in \N$, $K(X \restr m_k) \leq s\cdot m_k + o(m_k)$. We use the short description of $X \restr m_k$ to produce a short description of the Cantor covering $U \in \Phi_Q$ such that $x \in U$ and $|U| \leq 2^{-m_k}$. 
	
	There are at most three possibilities of $U \in \Phi_Q$ such that $2^{-(m_k+1)} < |U|$ and $x \in \U$, depending on $X \restr m_k$. As $Q$ is computable, we can find $<k,m>$ and hardcode $e \in \{0,1,2\}$ such that $\delta^{-1}(U) = <k,m+e>$.
	Therefore $K^A(\delta^{-1}(U)) \leq s\cdot m_k + o(m_k)$. As $|U| < 2^{-m_k}$, we have $K_r^A(x) \leq sr + o(r)$.
\end{proof}

Therefore for any computable Cantor covering $\Phi_Q$, we can characterise the Cantor covering dimension of a point as the limit infimum of the Kolmogorov complexities of its finite prefixes, along a subsequence $m_k$. The subsequence $m_k$ depends only on the terms in the cantor series $Q$. The proof follows from Lemma \ref{lem:KCChanging}, taking $A = \emptyset$.
\begin{theorem}\label{lem:KCCharnofPhidim}
	For any $x \in \X$, and any computable Cantor covering $\Phi_Q$
	generated by $Q = \{n_k\}_{k \in \N}$,
	$$\cdim_{\Phi_Q}(x) = \liminf\limits_{k\to\infty} \frac{K(X
		\restr m_k)}{m_k}$$ where $X$ is a binary
	expansion of $x$ and $m_k = \floor{\log_2(n_1 \cdot n_2 \dots
		n_k)}.$
\end{theorem}

From the point-to-set principle for $\Phi$-dimension (Theorem
\ref{thm:pspforPhidim}), we have the following point-to-set
principle for Cantor covering dimension.

\begin{theorem}\label{thm:pspforCantor}
	For all $\mathcal{F} \subseteq \X$ and for all 
	Cantor coverings $\Phi_Q$, $$\dim_{\Phi_Q}(\mathcal{F}) =
	 \min\limits_{A \subseteq \N} \; \sup\limits_{\x \in \mathcal{F}} \; \liminf\limits_{k\to\infty} \frac{K^A(X
	 	\restr m_k)}{m_k}$$ where $X$ is a binary
	 expansion of $x$ and $m_k = \floor{\log_2(n_1 \cdot n_2 \dots
	 	n_k)}.$
\end{theorem}
\begin{proof}
	The forward direction follows from proof of Theorem \ref{thm:pspforPhidim}. Into the oracle $A$, we also encode $h(k) = n_k$.
	Note that with oracle access to $A$, $\Phi_Q$ effectively becomes a set of computable cantor coverings. Therefore, we can apply the same techniques in the forward direction of proof of Lemma \ref{lem:KCChanging}, and shift from $K^A_r(x, \Phi)/r$ to $K^A(X
	\restr m_k)/m_k$. 
	
		The reverse direction of the inequlity proceeds along the same lines as proof of Theorem \ref{thm:pspforPhidim}. 
	For an $A \in \Sigma^\infty$, take an $s >$ $\sup\limits_{\x \in \mathcal{F}} $ $\; \liminf\limits_{k\to\infty} \frac{K^A(X
		\restr m_k)}{m_k}$. For all $x \in \F$, we have that for infinitely many $k \in \N$, $K^A(X \restr m_k) \leq s\cdot m_k$. 
	Define a $\Phi$-cover $\UU$ as follows. For each $k \in \N$, and for each
	$w \in \Sigma^{m_k}$, such that $K^A(w) \leq s \cdot m_k$, there are at most three sets $U \in \Phi_Q$ such that $X \restr m_k \in U$ and $ 2^{-{m_k+1}} < |U| \leq 2^{-m_k}$. Add them to $\UU$.
	It follows that $\F \subseteq \bigcup_{U \in \UU} U $.
	We also have $\sum_{U \in \UU_k} |U|^s \leq 3 \cdot \sum_{w \in \Sigma^*} 2^{-K^A(w)} \leq 3$. The last inequality follows from the Kraft inequality. Therefore $\cdim_{\Phi}(\F) \leq s$.
\end{proof}

Theorem \ref{lem:KCCharnofCdim} ensures that when the Kolmogorov complexities of any two $X, Y \in \Sigma^\infty$ align over all finite prefixes, their constructive dimensions become equal.
From Theorem \ref{lem:KCCharnofPhidim}, we get that when this happens, the constructive $\Phi$-dimensions for computable Cantor coverings
become equal.

\begin{lemma} \label{lem:KCsameMeansDimsSame}
	For any $x,y \in \X$, $A,B \subseteq \N$ and any class of
	computable Cantor coverings $\Phi$, if for all n, $|K^A(X \restr
	n) - K^B(Y \restr n)| = o(n)$, then $\cdim^A(x) =
	\cdim^B(y)$ and $\cdim^A_\Phi(x) = \cdim^B_\Phi(y)$. Here $X$ and $Y$ are the binary expansions of $x$ and $y$ respectively. 
\end{lemma} 

\section{Constructive Faithfulness of Cantor Coverings} \label{sec:CantorFaith}

In this section, we define the notion of faithfulness of a class of coverings towards constructive and Hausdorff Dimension. For Cantor coverings, we characterise the constructive dimension faithfulness using a log-limit condition of the terms appearing in the Cantor series expansion. 

Then, using the point-to-set principle and properties of
Kolmogorov complexity, we show that when the class of covers $\Phi$ is
generated by computable Cantor series expansions, the faithfulness at
the Hausdorff and constructive levels are equivalent notions.

\subsection{Faithfulness of family of coverings}
We will first see the definition of Hausdorff dimension faithfulness.
We then introduce the corresponding notion at the effective level,
which we call constructive dimension faithfulness.

A family of covering sets $\Phi$ is said to be \emph{faithful} with respect
to Hausdorff dimension if the $\Phi$ dimension of every set in the
space is the same as its Hausdorff dimension.

\begin{definition}
	A family of covering sets $\Phi$ over the space $\X$ is said to be \emph{faithful with respect to Hausdorff dimension} if for all
	$\F \subseteq \X$, $\dim_\Phi(\F) = \dim(\F)$.
\end{definition}


We extend the definition to the constructive level as well. A family of  covering sets with computable diameters $\Phi$ is defined to be \emph{faithful}
with respect to constructive dimension if the constructive $\Phi$
dimension of every set is the same as its constructive dimension.

\begin{definition}
	A family of  covering sets with computable diameters $\Phi$ is
	said to be \emph{faithful with respect to constructive
		dimension} if for all $\F \subseteq \X$, $\cdim_\Phi(\F) =
	\cdim(\F)$.
\end{definition}


The following lemma follows from Theorem \ref{thm:CDimSetfromPoint}.
It states that constructive dimension faithfulness can be equivalently stated in terms of preservation of constructive dimensions of points in the set. 

\begin{lemma}\label{lem:CPhidimSetPoint}
	A family of  covering sets $\Phi$ with computable diameters is
	{faithful with respect to constructive
		dimension} if and only if for all $\x \in \X$, $\cdim_\Phi(\x) =
	\cdim(\x)$.
	
\end{lemma} 


The following lemma states that the $\Phi$-dimension of a set is always greater than or equal to its Hausdorff dimension. Similarly, the constructive $\Phi-$dimension of a set is always greater than or equal to its constructive dimension. The proofs follow from the respective Definitions.

\begin{lemma} \label{lem:PhidimGeqdim}
	For any family of covering sets $\Phi$ over $\X$, for all $\F \subseteq \X$, 
	$\dim_\Phi(\F) \geq \dim(\F)$.
\end{lemma}

\begin{lemma} \label{lem:PhidimGeqCdim}
	For any family of covering sets with computable diameters $\Phi$ over $\X$, for all $\F \subseteq \X$, 
	$\cdim_\Phi(\F) \geq \cdim(\F)$.
\end{lemma}

Therefore  $\Phi$ is not faithful for Hausdorff dimension if and only if there exists an $\F \subset \X$ such that $\dim_{\Phi}(\F) > \dim(\F)$. Similarly, $\Phi$ is not faithful for constructive dimension if and only if there exists an $\F \subset \X$ such that $\cdim_{\Phi}(\F) > \cdim(\F)$.

\subsection{ Constructive faithfulness of Cantor coverings}

We show that the constructive dimension faithfulness of Cantor coverings can be determined using the terms $n_k$ in $Q$.

\begin{lemma} \label{lem:EffFaithCondn1}
	A family of computable Cantor coverings $\Phi_Q$ generated by $Q =
	\{n_k\}_{k \in \N}$ is faithful with respect to constructive
	dimension if  	
	\begin{align*}
		\lim_{k\rightarrow \infty} \frac{\log n_k}{\log n_1 \cdot n_2
			\dots n_{k-1}} = 0.
	\end{align*}
\end{lemma}

\begin{proof}
	Given any $x \in [0,1]$, we have that $\cdim_{\Phi_Q}(x)  = \liminf\limits_{k\to\infty} \frac{K(X
		\restr m_k)}{m_k}$ where $X$ is a binary
	expansion of $x$ and $m_k = \floor{\log_2(n_1 \cdot n_2 \dots
		n_k)}.$ Given $X \restr m_k$, and for any $\ell \in \N$, $X \restr m_k + \ell$ can be produced by supplying the additional $(m_{k} + \ell - m_k)$ bits of $X \restr m_k + \ell$. Therefore, $$K(X \restr m_{k} + \ell) \leq K(X \restr m_k) + \ell  + o(m_{k} + \ell).$$
	 Also $X \restr m_k$ can be computed given $X \restr m_k + \ell$ and $m_k$ by trimming to $m_k$ length. Therefore, $$K(X \restr m_k) \leq K(X \restr m_k + \ell) + O(\log m_k).$$
	
	 If $\lim_{k\rightarrow \infty} \frac{\log n_k}{\log n_1 \cdot n_2
		\dots n_{k-1}} = 0$, we have that $\lim_{k\rightarrow \infty} \frac{m_{k+1} - m_{k}}{m_{k}} = 0$. 
	
	Therefore, for all $\ell_k \leq m_{k+1} - m_k$, we have $\lim_{k\rightarrow \infty} \frac{ \ell_k }{m_{k}} = 0$ 
	
	Now,  \begin{align*}
		 \frac{K(X \restr m_{k} + \ell_k)}{m_{k} + \ell_k} - \frac{K(X \restr m_{k})}{m_{k}} &\leq  \frac{K(X \restr m_k) + \ell_k  + o(m_{k} + \ell_k)}{m_{k} + \ell_k} - \frac{K(X \restr m_{k})}{m_{k}} \\
		 &\leq \frac{\ell_k}{m_{k}} + \frac{ o(m_{k} + \ell_k)}{m_{k}+\ell_k}.
	\end{align*}

	And, \begin{align*}
		 \frac{K(X \restr m_{k})}{m_{k}} - \frac{K(X \restr m_{k} + \ell_k)}{m_{k} + \ell_k} &\leq  \frac{K(X \restr m_k + \ell_k) + o(m_k)}{m_{k}} - \frac{K(X \restr m_{k} + \ell_k)}{m_{k} + \ell_k} \\
		&\leq \bigg(\frac{K(X \restr m_{k} + \ell_k)}{m_{k} + \ell_k} \cdot \frac{\ell_k}{m_k} \bigg) + \frac{ o(m_{k} )}{m_{k}} \\
		&\leq \frac{\ell_k}{m_k} + \frac{ o(m_{k} )}{m_{k}}.
	\end{align*}

	From this, it follows that
		$\lim_{k \to \infty} \frac{K(X \restr m_{k} + \ell_k)}{m_{k} + \ell_k} - \frac{K(X \restr m_{k})}{m_{k}}  = 0$ when $\ell_k \leq m_{k+1} - m_k$.
	
	 Therefore,  for all $x \in [0,1]$, $\cdim_{\Phi_Q}(x) = \cdim(x)$.

\end{proof}

\begin{lemma} \label{thm:EffFaithCondn2}
	A family of computable Cantor coverings $\Phi_Q$ generated by $Q =
	\{n_k\}_{k \in \N}$ is faithful with respect to constructive
	dimension only if  	
	\begin{align*}
		\lim_{k\rightarrow \infty} \frac{\log n_k}{\log n_1 \cdot n_2
			\dots n_{k-1}} = 0.
	\end{align*}
\end{lemma}
\begin{proof}

	Fix a rational $0 < s < 1$. For a Cantor covering $\Phi_Q$ which does not satisfy the condition, we construct an $X \in \Sigma^\infty$ such that $\cdim(X) < s$ and $\cdim_{\Phi_Q}(X) = s$.

	 For all $k \in \N$, let $m_k = \floor{\log_2(n_1 \cdot n_2 \dots
		n_k)}.$ Let $s = p/q$ for some $p,q \in \N$. Let $R_s \in \Sigma^\infty$ be the sequence obtained by diluting $p$ bits of a Martin-L\"of random with $q-p$ bits of $0$'s. We see that $\lim_{n\to \infty} K(R_s \restr n)/n = s$. Let $R$ be a Martin-L\"of random independent of $R_s$.
	
	Let $\limsup_{k\rightarrow \infty} \frac{ m_{k+1} - m_{k}}{m_k} > \epsilon $ for some rational $\epsilon > 0$. Let $\{k_i\}_{i \in \N}$ be the corresponding indices such that for all $i \in \N$, ${ m_{k_i+1}} > (1 + \epsilon) \cdot {m_{k_i}}$. 
	

	We construct $X$ stage wise ensuring that the following requirements are satisfied. 
	
	\begin{enumerate}
		\item For all $k \in \N$, $K(X\restr m_{k}) = s \cdot m_{k} \pm o(m_k).$
		\item For some $\epsilon'>0$, for all $i \in \N$, 
		some $m_{k_i} < m_{k_i}'' < m_{{k_i}+1}$, $$K(X \restr m_{k_i}'') \leq (s - \epsilon') \cdot m_{k_i}''.$$
	\end{enumerate}
	
	Condition 1 guarantees that $\cdim_{\Phi}(X) = s$. Condition 2 guarantees that $\cdim(X) \leq s - \epsilon'$, thereby ensuring that $\Phi$ is not faithful for constructive dimension.

	\begin{figure}[h]
		\centering
		\begin{tikzpicture}  \label{figure}[ 
			>=stealth,
			thick,
			x=0.8cm,
			every node/.style={font=\small}
			]
			\usetikzlibrary{decorations.pathreplacing}
			
			\draw[->] (0,0) -- (12,0) node[right] {$ $};
			
			\coordinate (A) at (1,0);   
			\coordinate (B) at (3,0);   
			\coordinate (C) at (5.5,0); 
			\coordinate (D) at (8.5,0); 
			\coordinate (E) at (11,0);  
			
			\foreach \P/\lab in {A/$m_{k_{i-1}+1}$, B/$m_{k_i}$, C/$g_i$, D/$h_i$, E/$m_{k_i+1}$}
			\draw (\P) -- ++(0,0.25) node[above] {\lab};
			
			\foreach \P/\Q in {A/Ashift, B/Bshift, C/Cshift, D/Dshift, E/Eshift}
			\path (\P) ++(0,-0.4) coordinate (\Q);
			
			
			\draw[decorate,decoration={brace, mirror,amplitude=4pt}]
			(Ashift) -- (Bshift)
			node[midway,below=4pt] {$R_s$};
			
			\draw[decorate,decoration={brace, mirror,amplitude=4pt}]
			(Bshift) -- (Cshift)
			node[midway,below=4pt] {$0$};
			
			\draw[decorate,decoration={brace, mirror,amplitude=4pt}]
			(Cshift) -- (Dshift)
			node[midway,below=4pt] {$R$ };
			
			\draw[decorate,decoration={brace, mirror,amplitude=4pt}]
			(Dshift) -- (Eshift)
			node[midway,below=4pt] {$R_s$};
			
		\end{tikzpicture}
		\caption{Filling of $X$ during stage $i$.}
	\end{figure}
	
	We provide the stage wise construction of $X$. We start with $X_0 = \lambda$. We describe the extension of $X$ at the $i^{th}$ stage. At the beginning of this stage, the bits of $X$ till $m_{k_{i-1} + 1} $ have already been fixed.
	 Fill $X[m_{k_{i-1} + 1}  , m_{k_i} - 1]$ using the next bits of $R_s$. 	
	This ensures that at the start of stage $i$, Condition 1 is met for all $k \in k_{i-1} + 1, \dots, k_{i}$.  
	

	We are given that  ${m_{k_i+1}} > (1 +  \epsilon) {m_{k_i}}$. Let $h_i = \floor{(1 + \epsilon) m_{k_i}}$. 
	
	We choose a value $g_i \in \N$ (which we fix soon) such that $m_{k_i} < g_i < h_i$. Fill $X[m_{k_i} : g_i]$ with $0$ and $X[g_i : h_i]$ from the next bits of the Martin-L\"of random source $R$. 
	
	
	We first analyse $K(X \restr h_i)$.
	
	Using the symmetry of information of Kolmogorov complexity (see \cite{LiVitanyi}), upto $\log(|x|.|y|)$ error, $K(x,y) = K(x) + K(y|x)$. If $x$ and $y$ are independent, we have $K(x,y) = K(x) + K(y)$.

	Therefore upto $o(m_{k_i})$ terms, 
	\begin{align*}
		K(X \restr h_i) &= K(X \restr m_{k_{i}}) +  K(X[m_{k_i} : g_i]) + K(X[g_i : h_i])\\
		& = h_i - g_i + s \cdot m_{k_i}.
	\end{align*}
	We choose $g_i$ such that $h_i - g_i = s (h_i - m_{k_i})$. This ensures that $K(X \restr h_i) = s \cdot h_i$. 
	
	Finally, we set $X[h_i : m_{k + 1}]$ using the subsequent bits of $R_s$, ensuring that Condition 1 is satisfied at the end of stage $i$.
	
	Analysing the Kolmogorov complexity at $g_i$, we have
	$$K(X \restr g_i) \leq s\cdot m_{k_i} + o(g_i).$$
	
	Substituting the value of $g_i$ and ignoring $o(.)$ terms, we get
	\begin{align*}
		\frac{K(X \restr g_i)}{g_i} &\leq \frac{s \cdot m_{k_i}}{h_i(1-s) + s \cdot m_{k_i}} \\
		&= \frac{s \cdot m_{k_i}}{m_{k_i}(s + (1 + \epsilon) (1- s))}\\
		&= \frac{s}{1  + \epsilon (1 - s)}.
	\end{align*}

	Since $\epsilon(1-s) > 0$,  Condition 2 is satisfied.
	
	Therefore, we have $\cdim(X) < s$ and $\cdim_{\Phi_Q}(X) = s$. 
\end{proof}

\begin{theorem} \label{thm:EffFaithCondn}
	A family of computable Cantor coverings $\Phi_Q$ generated by $Q =
	\{n_k\}_{k \in \N}$ is faithful with respect to constructive
	dimension if and only if 	
	\begin{align}
		\label{log_limit_condition}
		\lim_{k\rightarrow \infty} \frac{\log n_k}{\log n_1 \cdot n_2
			\dots n_{k-1}} = 0.
	\end{align}
\end{theorem}

The Cantor series expansion is a generalization of the
base-$b$ representation, which is the special case when $n_k = b$ for
all $k \in \N$. That is $Q_b = \{b\}_{n\in\N}$. Since the condition in
Theorem \ref{thm:EffFaithCondn} is satisfied by $Q_b$ for any $b \in
\N$, we have the following result by Hitchcock and Mayordomo about the
base invariance of constructive dimension. 

\begin{corollary}[Hitchcock and Mayordomo \cite{HitchMayor13}]
	For any $\x \in [0,1]$ and $k,l \in \N \setminus \{1\}$,
	$\cdim_{(k)}(\x) = \cdim_{(l)}(\x).$ where $\cdim_{(k)}(\x)$
	represents the constructive dimension of $\x$ with respect to
	its base-$k$ representation.
\end{corollary}

Note that condition (\ref{log_limit_condition}) classifies the Cantor series expansions on the basis of constructive dimension faithfulness. As an example, when $n_k = {2^k}$, condition (\ref{log_limit_condition}) holds, and therefore $Q = \{2^k\}_{k \in \N}$ is faithful for constructive dimension. However,  when $n_k = {2^{2^k}}$,  condition (\ref{log_limit_condition}) does not hold, and therefore $Q = \{2^{2^k}\}_{k \in \N}$ is not faithful for constructive dimension.

\section{Hausdorff dimension faithfulness of Cantor coverings}\label{sec:HausdFaith}

We give an alternate, information theoretic proof of the result by Albeverio, Ivanenko, Lebid and Torbin \cite{Albeverio2020} that the Hausdorff dimension faithfulness of Cantor coverings can be determined using the same log-limit condition over the terms $n_k$ in $Q$.

We first show that if the log-limit condition holds for $Q$, then $\Phi_Q$ is faithful for Hausdoff dimension.

\begin{lemma}[Albeverio, Ivanenko, Lebid and Torbin  \cite{Albeverio2020}] \label{lem:AlbIvankio1}
	
	A family of  Cantor coverings $\Phi_Q$ generated by $Q = \{n_k\}_{k \in \N}$ is faithful with respect to Hausdorff dimension if 
	$$\lim\limits_{k\rightarrow \infty} \frac{\log n_k}{\log n_1 \cdot n_2
		\dots n_{k-1}} = 0.$$
\end{lemma}
\begin{proof}
	Suppose the log-limit condition $\lim\limits_{k\rightarrow \infty} \frac{\log n_k}{\log n_1 \cdot n_2
		\dots n_{k-1}} = 0$ holds. 
	Using the relativised version of proof of Lemma \ref{lem:EffFaithCondn1},  
	we have that for all $X \in \Sigma^\infty$ and $A \in \Sigma^\infty$, for any sequence  $\ell_k \leq m_{k+1} - m_k$, 
	$$\lim_{k \to \infty} \frac{K^A(X \restr m_{k} + \ell_k)}{m_{k} + \ell_k} - \frac{K^A(X \restr m_{k})}{m_{k}}  = 0.$$
	
	Therefore, we have that $\liminf\limits_{k \to \infty}{\frac{K^A(X \restr m_{k})}{m_{k}}} = \liminf\limits_{n \to \infty} \frac{K^A(X \restr n )}{n}$.	
	
	Using the point-to-set principles for Hausdorff dimension and $\Phi_Q$ (Corollary \ref{thm:PSP} and Theorem \ref{thm:pspforCantor}), we have that for all $\F \subseteq \X$, $\dim(\F) = \dim_{\Phi_Q}(\F)$.

\end{proof}

	Now suppose the log-limit condition  does not hold for $Q$. For any rational $0 < s < 1$, we show the existence of a set $\F \subseteq [0,1]$ such that $\dim_{\Phi_Q}(\F) = s$ and $\dim(\F) < s$.

\begin{definition}
As in proof of Lemma \ref{thm:EffFaithCondn2}, since the log-limit condition does not hold, for some rational $\epsilon > 0$, there exists a sequence of indices $\{k_i\}_{i \in \N}$ such that for all $i \in \N$, ${ m_{k_i+1}} > (1 + \epsilon) \cdot {m_{k_i}}$.
\begin{itemize}	
\item Let $h_i = \floor{(1 + \epsilon) m_{k_i}}$. 
\item Let $g_i = h_i + s (h_i - m_{k_i})$. 
\end{itemize}
\end{definition}

\begin{definition}
Let $0 < s < 1$ be a rational such that $s = p/q$ for the smallest $p,q \in \N$.	
Let $\F_s\subseteq[0,1]$ be the set of all points $x$ that satisfy the following conditions.  Let $X$ be a binary expansion of $x$. 
\begin{enumerate}
	\item For every position $j$, except those lying between $m_{k_i}$ and $h_i$, if $j \mod q  \geq p$, we have $X[j] = 0$.  
	\item Additionaly, at all positions $j$ such that  $m_{k_i} \leq j \leq g_i$, $X[j] = 0$.
\end{enumerate}
\end{definition}

In any $X \in \F$, except between indices that fall between, $m_{k_i}$ and $h_i$, $0$ is fixed at $q - p$ out of $q$ bits (Condition 1). Therefore, we only need $s$ fraction of bits to specify the bits in these positions. Even at indices between $m_{k_i}$ and $h_i$, positions between $m_{k_i}$ and $g_i$ are $0$'s. This along with the choice of $g_i$, gives that only $s$ fraction of the bits are needed to represent $K(X \restr m_k)$, for all $k$. 

\begin{lemma} \label{lem:dimQF_slts}
	$\dim_{\Phi_Q}(\F_s) \leq s$
\end{lemma}
\begin{proof}
	Assume some $k \in \N$, $K(X \restr m_k) \leq s \cdot m_k + o(m_k)$. In first case, $k$ is such that $k = k_i$ for some $i \in \N$. We have $K(X \restr m_{k_i + 1}) \leq K(X \restr m_{k_i}) + K(X[m_{k_i} : h_i]) + K(X[h_i : m_{k_i+1}]) + o(m_{{k_i}+1}) \leq s \cdot m_{k_i+1} + o(m_{k_i+1})$. 
	
	Otherwise,  $k \neq k_i$ for all $i$. In this second case, again $K(X \restr m_{k+1}) \leq K(X \restr m_k) + K(X[m_k : m_{k+1}]) + o(m_{k+1}) \leq s \cdot m_{k+1} + o(m_{k+1})$. Therefore, we have that for all $X \in \N$, for all $k \in \N$, $K(X \restr m_k) \leq s \cdot m_k  +o(m_k)$.
	
	From the point-to-set principle  for $\Phi_Q$-dimension (Theorem \ref{thm:pspforCantor}), we have $\cdim_{\Phi_Q}(\F) \leq s$.
\end{proof}

For the reverse direction, take an oracle $A$, and let $R$ be a Martin-L\"of random relative to $A$. We consider a sequence $X$ in which bits outside $I$ in are filled from $R$.
So, even with access oracle $A$, bits outside $I$ have maximal information density. 

Using symmetry of information and a case-wise analysis, we show that the positions $X \restr m_k$ have a density of information $s$. This along with the point-to-set principle for $\Phi_Q$ shows that $\dim_{\Phi_Q}(\F_s) \geq s$.

\begin{lemma}\label{lem:dimQF_sgts}
		$\dim_{\Phi_Q}(\F_s) \geq s$
\end{lemma}
\begin{proof}
	
	For any $A \in \Sigma^\infty$, let $R$ be a Martin-L\"of random relative to $A$. Let $X$ be the sequence for which the positions unfixed with 0 (indices outside $I$), are filled from bits of $R$. That is
	$X \in \F$ is such that $X[j] = R[j]$ for all $j \not\in I$ and $X[j] = 0$ for all $j \in I$.

		We first show that
		for all $k \in \N$, $K^A(X_A\restr m_{k}) = s \cdot m_{k} \pm o(m_k).$
		
	From the symmetry of information of Kolmogorov complexity \cite{Downey10}, we have for any $x, y \in \Sigma^*$, upto $o(|x| + |y|)$ terms, $K^A(x,y) = K^A(x) + K^A(x | y)$.
	If $x$ and $y$ are disjoint substrings of $X$, by design of $X$, they are independent, and so upto $o(|x| + |y|)$ terms, $K^A(x,y) = K^A(x) + K^A(y)$.

	Assume that for some $k \in \N$, $K^A(X_ \restr m_k) = s\cdot m_k + o(m_k)$.  In the first case, $k$ is such that $k = k_i$ for some $i \in \N$. Using symmetry of information, we have 
	$K^A(X \restr h_i) = K^A(X \restr m_{k_i}) + K^A(X[m_{k_i} : h_i]) = s \cdot m_{k_i} + h_i - g_i = s \cdot h_i$. Again,  upto $o(m_{k_i + 1})$ terms, $K^A(X \restr m_{k_i + 1}) = K^A(X \restr h_i) + K^A(X [h_i : m_{k_i + 1}]) = s \cdot h_i + s (m_{k_i + 1} - h_i) = s \cdot m_{k_i + 1}$. 
	
	In the second case, $k \neq k_i$ for all $i$. In this case, we directly have $K^A(X \restr m_{k+1}) = K^A(X \restr m_k) + K(X[m_k : m_k+1]) = s \cdot m_{k} + s (m_{k+1} - m_k) = s \cdot m_{k+1}$.
	
	Therefore, we have that for all $A \in \Sigma^\infty$, there exists a $X \in \F$ such that for all $k \in \N$, $K^A(X \restr m_k) = s \cdot m_k + o(m_k)$. 
	
	From the point-to-set principle  for $\Phi_Q$-dimension (Theorem \ref{thm:pspforCantor}), we have $\cdim_{\Phi_Q}(\F) \geq s$.
\end{proof}

	We now show that due to the $0$'s between $m_{k_i}$ and $g_i$, the Kolmogorov complexity of $X \in \F_s$ dips strictly below $s$, at $g_i$. Applying the point-to-set principle for Hasudroff dimension, we show that Hausdorff dimension of $\F_s$ is strictly below $s$.
	
	\begin{lemma}\label{lem:dimF_slts}
		$\dim(\F_s) < s$
	\end{lemma}
	\begin{proof}
		In the proof of Lemma \ref{lem:dimQF_slts}, we have established that for  all $X \in \F$ and $k \in \N$, $K(X \restr m_k) \leq s \cdot m_k +o(m_k)$. As $X[m_{k_i} : g_i]$, are all $0$, we have
		$K(X \restr g_i) \leq s\cdot m_{k_i} + o(g_i).$
		
		Just as in proof of Lemma \ref{thm:EffFaithCondn2}, substituting the value of $g_i$ and ignoring $o(.)$ terms, we get
		\begin{align*}
			\frac{K(X \restr g_i)}{g_i} &\leq \frac{s \cdot m_{k_i}}{h_i(1-s) + s \cdot m_{k_i}} \\
			&= \frac{s \cdot m_{k_i}}{m_{k_i}(s + (1 + \epsilon) (1- s))}\\
			&= \frac{s}{1  + \epsilon (1 - s)}.
		\end{align*}
		
		Since $0<s<1$, we have that for all $i \in \N$ and $X \in \F_s$, 
		$$\frac{K(X \restr g_i)}{g_i} < s.$$
		From the point-to-set principle for Hausdorff dimension (Corollary \ref{thm:PSP}),  $\dim(\F_s) < s$.
	\end{proof}

	From Lemma \ref{lem:AlbIvankio1}, we have that if the log-limit condition holds for $Q$, $\Phi_Q$ is faithful for Hausdorff dimension. From Lemma \ref{lem:dimQF_sgts} and Lemma \ref{lem:dimF_slts}, we have that if the log-limit condition does not hold for $Q$, $\Phi_Q$ is not faithful for Hausdorff dimension.
	 
	\begin{theorem}[Albeverio, Ivanenko, Lebid and Torbin  \cite{Albeverio2020}] \label{thm:AlbIvankio2}
		
		A family of  Cantor coverings $\Phi_Q$ generated by $Q = \{n_k\}_{k \in \N}$ is faithful with respect to Hausdorff dimension if and only if $$
		\lim\limits_{k\rightarrow \infty} \frac{\log n_k}{\log n_1 \cdot n_2
			\dots n_{k-1}} = 0.$$
	\end{theorem}

\section{Equivalence of Faithfulness of Coverings at Constructive and Hausdorff Levels} \label{sec:EqlFaith}

In this section, we investigate whether the notions of constructive faithfulness and classical faithfulness are equivalent for all covers $\Phi$ with computable diameters over the Euclidean space $\R^n$. We formulate this as a conjecture.

\begin{conjecture} \label{conj}
	Let $\Phi$ be a family of covering sets with computable diameters over $\R^n$. 
	
	Then $\Phi$ is  faithful for constructive dimension if and only if $\Phi$ is faithful for Hausdorff dimension.
\end{conjecture}

 Existing evidence supports this possibility. 
 For instance, constructive Cantor series dimension the same log-limit condition characterizes faithfulness at both the classical and constructive levels (Theorem \ref{thm:EffFaithCondn} and Theorem \ref{thm:AlbIvankio2}).
 Also, continued fraction covers are not faithful for constructive dimension \cite{Nandakumar2023b}, or Hausdorff dimension \cite{PeresTorbin}. 
 
 Towards proving such a general statement, we investigate the structural reasons that might underly such an equivalence. The rest of the paper is dedicated to proving the equivalence of faithfulness for constructive Cantor coverings through an abstract argument, independent of the log-limit characterization of faithfulness.

\subsection{Kolmogorov Complexity Construction} \label{sec:KCConstrn}
We first give a technical construction which is crucial in proving the results in this section. This construction may also be of independent interest.

We show that given an infinite sequence $X$ and an oracle $A$, for any oracle $B$, there exists a sequence $Y$ whose relativised Kolmogorov complexity (of prefixes) with respect to $B$ is similar to the relativised Kolmogorov complexity (of prefixes) of $X$ with respect to $A$. 

\begin{theorem}\label{thm:KCChasinglemma}
	For all $X\in\Sigma^\infty$ and $A\in\Sigma^\infty$, for every $B\in\Sigma^\infty$ there exists $Y\in\Sigma^\infty$ such that for all $n\in\mathbb N$,
	\[
	\bigl|K^{A}(X\upharpoonright n)-K^{B}(Y\upharpoonright n)\bigr|=o(n)
	\qquad\text{and}\qquad
	\cdim^{B}(Y)=\cdim(Y).
	\]
\end{theorem}

\begin{proof}
	Let $k(n):=K^{A}(X\restr n)$. We construct $Y$ in stages as a concatenation of blocks.
	Set $\ell_m:=m$ and $n_m:=\sum_{i=1}^m \ell_i=m(m+1)/2$.
	Inductively define strings $\tau_m\in\Sigma^{n_m}$ by $\tau_0=\lambda$ and $\tau_m=\tau_{m-1}\rho_m$ for suitable $\rho_m\in\Sigma^{\ell_m}$; finally let $Y=\lim_{m\to\infty}\tau_m$.
	
	Fix $m\ge 1$ and put $\Delta_m:=k(n_m)-k(n_{m-1})$. Since $n_m=n_{m-1}+\ell_m$, we have
	\[
	k(n_m)=K^{A}(X\restr(n_{m-1}+\ell_m))\le K^{A}(X\restr n_{m-1})+\ell_m+O(\log \ell_m),
	\]
	obtained by describing $X\restr n_{m-1}$ and then appending the next $\ell_m$ bits literally (plus a self-delimiting code for $\ell_m$). Thus $\Delta_m\le \ell_m+O(\log \ell_m)$.
	Also $k(n_{m-1})\le k(n_m)+O(1)$ because $X\restr n_{m-1}$ is uniformly computable from $X\restr n_m$ (truncate to length $n_{m-1}$, which is computable from the output length $n_m$), hence $\Delta_m\ge -O(1)$.
	Define
	\[
	t_m:=\max\{0,\min\{\Delta_m,\ell_m-1\}\}\in[0,\ell_m-1],
	\]
	so that $\Delta_m-t_m=O(\log\ell_m)$.
	
	Choose $u_m\in\Sigma^{t_m}$ such that $K^{B}(u_m\mid\tau_{m-1})\ge t_m-O(1)$.
	This exists by counting: among the $2^{t_m}$ strings of length $t_m$, fewer than $2^{t_m-1}$ can satisfy
	$K^{B}(\cdot\mid\tau_{m-1})\le t_m-2$ since there are too few prefix-free programs of length $\le t_m-2$.
	Now define
	\[
	\rho_m:=u_m\,1\,0^{\ell_m-t_m-1}\in\Sigma^{\ell_m}.
	\]
	The marker $1$ followed only by zeros makes $u_m$ uniformly computable from $\rho_m$ alone (scan from the end to the last $1$). Therefore
	\[
	K^{B}(\rho_m\mid\tau_{m-1}) = t_m+O(1).
	\]
	Moreover, oracle access cannot increase complexity by more than an additive constant, so
	$K(u_m\mid\tau_{m-1})\ge K^{B}(u_m\mid\tau_{m-1})-O(1)\ge t_m-O(1)$,
	and since $u_m$ is computable from $\rho_m$ we also obtain
	\[
	K(\rho_m\mid\tau_{m-1}) = t_m+O(1).
	\]
	
	Let $c_{sym}^A$ and $c_{sym}^B$ be the constants from the asymptotic error term in the symmetry of information of Kolmogorov complexity with oracles $A$ and $B$ respectively (see Theorem 3.10.2 in \cite{Downey10}).
	Applying symmetry of information to the concatenation $\tau_m=\tau_{m-1}\rho_m$ gives
	\[
	K^{B}(\tau_m)=K^{B}(\tau_{m-1})+K^{B}(\rho_m\mid\tau_{m-1})+O(\log \ell_m),
	\]
	where the hidden constant in $O(\log\ell_m)$ is bounded by $c_{sym}^B$. Hence
	\[
	K^{B}(\tau_m)=K^{B}(\tau_{m-1})+t_m+O(\log \ell_m).
	\]
	Define $E_m:=K^{B}(\tau_m)-k(n_m)$. Using $k(n_m)=k(n_{m-1})+\Delta_m$ and $\Delta_m-t_m=O(\log\ell_m)$, we get
	\[
	E_m = E_{m-1}+O(\log\ell_m).
	\]
	Therefore $|E_m|=O\!\left(\sum_{i=1}^m \log \ell_i\right)=O(m\log m)$.
	
	Now fix any $n\in\N$ and choose $m$ such that $n_{m-1}\le n\le n_m$ (so $n-n_{m-1}\le \ell_m$).
	By the same argument used above, for every oracle $C$,
	\[
	K^{C}(Z\restr(n_{m-1}+(n-n_{m-1})))\le K^{C}(Z\restr n_{m-1})+(n-n_{m-1})+O(\log (n-n_{m-1})),
	\]
	and in particular both $K^{B}(Y\restr n)$ and $k(n)$ differ from their values at $n_{m-1}$ by at most $O(\ell_m)$.
	Hence
	\[
	\bigl|K^{B}(Y\restr n)-k(n)\bigr|
	\le \bigl|K^{B}(Y\restr n_{m-1})-k(n_{m-1})\bigr|+O(\ell_m)
	= |E_{m-1}|+O(\ell_m)
	= O(m\log m)+O(m).
	\]
	Since $n_m=m(m+1)/2$, we have $m\le \sqrt{2n_m}\le 2\sqrt{n}$ for all large $n$, and therefore
	$O(m\log m)=O(\sqrt{n}\log n)=o(n)$. Thus $\bigl|K^{A}(X\restr n)-K^{B}(Y\restr n)\bigr|=o(n)$.
	
	Let $c_{sym}$ be the constant from the asymptotic error term in the symmetry of information for non-relativised Kolmogorov complexity (see Theorem 3.10.2 in \cite{Downey10}).
	Applying symmetry of information (once unrelativised and once relativised to $B$) to $\tau_m=\tau_{m-1}\rho_m$ yields
	\[
	K(\tau_m)=K(\tau_{m-1})+K(\rho_m\mid\tau_{m-1})+O(\log \ell_m),
	\qquad
	K^{B}(\tau_m)=K^{B}(\tau_{m-1})+K^{B}(\rho_m\mid\tau_{m-1})+O(\log \ell_m),
	\]
	where the hidden constants in the two $O(\log\ell_m)$ terms are bounded by $c_{sym}$ and $c_{sym}^B$, respectively.
	Subtracting and using $K(\rho_m\mid\tau_{m-1})=t_m+O(1)$ and $K^{B}(\rho_m\mid\tau_{m-1})=t_m+O(1)$ gives
	\[
	K(\tau_m)-K^{B}(\tau_m)=\bigl(K(\tau_{m-1})-K^{B}(\tau_{m-1})\bigr)+O(\log \ell_m),
	\]
	and hence $\bigl|K(\tau_m)-K^{B}(\tau_m)\bigr|=O(m\log m)$.
	
	For an arbitrary $n$ with $n_{m-1}\le n\le n_m$, the same extension estimate implies that both
	$K(Y\restr n)$ and $K^{B}(Y\restr n)$ differ from their values at $n_{m-1}$ by at most $O(\ell_m)$, so
	\begin{align*}
		\bigl|K(Y\restr n)-K^{B}(Y\restr n)\bigr|
		&\le \bigl|K(Y\restr n_{m-1})-K^{B}(Y\restr n_{m-1})\bigr|+O(\ell_m) \\
		&= O(m\log m)+O(m) \\
		&= O(\sqrt{n}\log n)
		= o(n).
	\end{align*}
	Dividing by $n$ and taking $\liminf$ yields $\cdim(Y)=\cdim^{B}(Y)$.
\end{proof}

\subsection{Equivalence of Faithfulness for Constructive Cantor Coverings}\label{subsec:EqlCantor} 

We first show that if a class of computable Cantor coverings $\Phi$ is faithful with respect to constructive
dimension, then $\Phi$ is also faithful with respect to Hausdorff
dimension.

\begin{lemma} \label{lem:Cantor1}
	For any class of computable Cantor coverings $\Phi$, if for all $\F \subseteq
	\X$, $\cdim(\F) = \cdim_\Phi(\F)$, then for all $\F \subseteq \X$,
	$\dim(\F) = \dim_\Phi(\F)$.
\end{lemma}

\begin{proof} 
	
	We first show that if $\cdim(\F) = \cdim_\Phi(\F)$ for all $\F
	\subseteq \X$ , then for all $A \subseteq \N$ and $\x \in [0,1]
	\setminus \Q$, $\cdim^A(\x) = \cdim^A_\Phi(\x)$.  
	
	Let $B = \emptyset$. From Theorem \ref{thm:KCChasinglemma}, we have
	that for all  $\x \in [0,1]$, and $A \subseteq \N$, there
	exists a $y \in [0,1]$,  such that for all $n \in \N$,
	$|K^A(X \restr n) - K(Y \restr n)| = o(n)$. 
	Here $X$ and $Y$ are binary expansions of $x$ and $y$ respectively.
	Therefore from Lemma \ref{lem:KCsameMeansDimsSame}, we have
	that $\cdim^A(x) = \cdim(y)$ and $\cdim^A_\Phi(x) =
	\cdim_{\Phi}(y)$.
	
	Since $\Phi$ is faithful with respect to constructive
	dimension,   $\cdim(y) = \cdim_\Phi(y)$. 
	Therefore we have that $\cdim^A(x) = \cdim_\Phi^A(x)$.
	
	Let $\F \subset \X$ be arbitrary. From the
	point-to-set principle for dimension of Cantor coverings
	(Corollary \ref{thm:pspforCantor}), 
	\begin{align*}
		\dim_\Phi(\mathcal{F}) &= \min\limits_{A \subseteq \N} \;
		\sup\limits_{\x \in \mathcal{F}} \; \cdim^A_\Phi(\x) = \min\limits_{A \subseteq \N} \; \sup\limits_{\x \in \mathcal{F}} \; \cdim^A(\x)= \dim(\F).
	\end{align*}

	The last equality follows from the point-to-set principle for
	Hausdorff dimension (Corollary \ref{thm:PSP}). 
\end{proof}

To prove the converse, we require the construction of the set $\mathcal{I}_s$ that contains all points in $\X$ having constructive dimension equal to $s$.

\begin{definition} \label{Def:TBD}
	Given $s \in [0, \infty)$, define $ \mathcal{I}_s = \{\x \in \X \;\vert\; \cdim(\x) = s\}.$
\end{definition}

Lutz \cite{Lutz2003b} showed that the Hausdorff dimension of $\mathcal{I}_s$ is equal to $s$. We provide a simple alternate proof of this using the point-to-set principle.

\begin{lemma}[Lutz \cite{Lutz2003b}] \label{lem:dimEsequalss}
	For all $s \in [0, 1]$, $\dim(\mathcal{I}_s) = s.$
\end{lemma}

\begin{proof}
	Since for all $\x \in \mathcal{I}_s$, by definition $\cdim(\x) =
	s$, from  \cite{Lutz2003b}, we have  $\cdim(I_s) = \sup_{x \in I_s} \cdim(x)$, and so we have
	$\cdim(\mathcal{I}_s) = s$. From \cite{Lutz2003b}, we have  $\dim(\mathcal{I}_s)
	\leq \cdim(\mathcal{I}_s)$, from which it follows that
	$\dim(\mathcal{I}_s) \leq s$.
	
	Consider any $\x \in \mathcal{I}_s$. We have that $\cdim(\x) = s$. 
	Let $A =
	\emptyset$. From Theorem \ref{thm:KCChasinglemma},  for all $B \subseteq \N$, there
	exists a $y \in [0,1]$ such that for all $n \in \N$,  $|K^A(X
	\restr n) - K^B(Y \restr n)| = o(n)$ having
	$\cdim^B(Y) = \cdim(Y)$. 	Here $X$ and $Y$ are binary expansions of $x$ and $y$ respectively.
	From Lemma \ref{lem:KCsameMeansDimsSame}, we have 
	$\cdim^B(y) = \cdim^A(x) = \cdim(x) = s$. 
	
	Therefore,  for all $B \subseteq \N$, there exits a $y \in \X$ such that  $\cdim(y) = \cdim^B(y) = s$.  Hence $Y \in \mathcal{I}_s$ and therefore from the point-to-set principle for Hausdorff dimension (Corollary \ref{thm:PSP}), we have  $\dim(\mathcal{I}_s) = \min\limits_{B \subseteq \N} \sup\limits_{y \in \mathcal{I}_s} \cdim^B(y) \geq s$. 
\end{proof}

We now show that if a class of computable Cantor coverings $\Phi$ is faithful with respect to Hausdorff dimension, then $\Phi$ is also faithful with respect to constructive dimension.

\begin{lemma} \label{lem:Cantor2}
	For any class of computable Cantor coverings $\Phi$, if for all $\F \subseteq \X$, $\dim(\F) = \dim_\Phi(\F)$, 
	then for all $\F \subseteq \X$, $\cdim(\F) = \cdim_\Phi(\F)$.
\end{lemma}

\begin{proof}
	Given an $\x \in \X$ having $\cdim(\x) = s$, consider the set $\mathcal{I}_s$ from Definition \ref{Def:TBD}. From Lemma \ref{lem:dimEsequalss}, we have $\dim(\mathcal{I}_s) = s$. Since $\Phi$ is faithful with respect to Hausdorff dimension,  $\dim_\Phi(\mathcal{I}_s) = s$. 
	
	We now   show that for any $\x \in \X$ having $\cdim = s$ (or equivalently, for any $x \in \mathcal{I}_s$), $\cdim_\Phi(\x) \leq s$.
	This along with Lemma \ref{lem:PhidimGeqCdim} shows that $\cdim(\x) = \cdim_\Phi(\x)$. From Lemma \ref{lem:CPhidimSetPoint}, we have that for all $\F \subseteq \X$, $\cdim(\F) = \cdim_{\Phi}(\F)$.

	We first show that for all $B \subseteq \N$, there exists a $y \in \mathcal{I}_s$  having $\cdim^B_\Phi(y) = \cdim_\Phi(x)$.
	Let $A = \emptyset$. From Theorem \ref{thm:KCChasinglemma}, we have
	that for all  $x \in [0,1]$, and $B \subseteq \N$, there
	exists a $y \in [0,1]$,  such that for all $n \in \N$,
	$|K^A(X \restr n) - K^B(Y \restr n)| = o(n)$. Here $X$ and $Y$ are binary expansions of $x$ and $y$ respectively.
	Setting $A = \emptyset$ and using Lemma \ref{lem:KCsameMeansDimsSame}, we have  $\cdim_\Phi(x) = \cdim^B_\Phi(y)$.
	
	We now show that $\dim_\Phi(\mathcal{I}_s) \geq \cdim_\Phi(X)$.
	From the point-to-set principle for dimension of Cantor coverings (Corollary \ref{thm:pspforCantor}), we have  $	\dim_\Phi(\mathcal{I}_s) = \min\limits_{B \subseteq \N} \; \sup\limits_{y \in \mathcal{I}_s} \; \cdim^B_\Phi(y)$. From the argument given above, for all $B \subseteq \N$, there exists a $y \in \mathcal{I}_s$ having $\cdim^B_\Phi(y) = \cdim_{\Phi}(x)$.
	So it follows that $\dim_{\Phi}(\mathcal{I}_s) \geq \cdim_{\Phi}(x)$.  Since we have already shown that $\dim_\Phi(\mathcal{I}_s) = s$, it follows that $\cdim_\Phi(x) \leq s$.
\end{proof}

Therefore, we have the following theorem which states that for the classes of Cantor coverings $\Phi$, faithfulness with respect to Hausdorff and constructive dimensions are equivalent notions.

\begin{theorem} \label{thm:equivalenceOfFaithfulness}
	For any class of computable Cantor coverings $\Phi$, 
	
	$$\forall \F \subseteq \X \;;\; \dim(\F) = \dim_\Phi(\F) \iff \forall \F \subseteq \X \;;\; \cdim(\F) = \cdim_\Phi(\F).$$ 
\end{theorem}

\section{Open Problems} \label{sec:Conclusion}

The following are some problems that remain open.

\begin{enumerate}
	\item Does equivalence of faithfulness (Theorem \ref{thm:equivalenceOfFaithfulness}) hold for a general $\Phi$ ?  See Conjecture \ref{conj}.
	
	\item[] We can also study a packing dimension analogue of faithfulness, and faithfulness of effective strong dimension \cite{strongdim}.   
	
	\item Does the log-limit condition characterise packing dimension faithfulness or effective strong dimension faithfulness of Cantor coverings ?
	
	\item Is there any relationship between faithfulness of Hausdorff dimension and packing dimension?
	
	\item Is there any relationship between faithfulness of constructive dimension and constructive strong dimension? 
\end{enumerate}

\section*{Acknowledgements}
We thank Laurent Bienvenu, Bill Mance and the anonymous reviewers of previous versions of this work for their comments and helpful suggestions.

\bibliography{main2,main_jabref_shared2}

\end{document}

%% file: commands2.tex
\newcommand{\PROOF}{\begin{proof}}

\newcommand{\QED}{\end{proof}}

\newcommand{\floor}[1]{ \left\lfloor #1 \right\rfloor }

\newcommand{\restr}{\upharpoonright}

\newcommand{\N}{\mathbb{N}}

\newcommand{\Q}{\mathbb{Q}}
\newcommand{\R}{\mathbb{R}}

\newcommand{\cdim}{\mathrm{cdim}}

\renewcommand{\dim}{{\mathrm{dim}}}

\renewcommand{\P}{\ensuremath{{\mathrm P}}}

\newcommand{\F}{{\mathcal{F}}}
\newcommand{\X}{{\mathbb{X}}}
\newcommand{\UU}{{\mathcal{W}}}

\newcommand{\Rplus}{[0,\infty)}
\renewcommand{\F}{\mathcal{F}}
\newcommand{\U}{U}

\newcommand{\defPhi}{\bigcup_{n \in \N}\{U_i^n\}_{i \in
		\N}}
	
\newcommand{\x}{x}	

\newcommand{\Rho}{\mathcal{P}}

\usepackage[inline]{enumitem}